\documentclass[oneside,reqno,10pt,letterpaper]{amsart}
\usepackage{amsmath,amsfonts,amssymb,fancyhdr}
\usepackage[dvips,left=1.3in,right=1.3in,top=1.3in, bottom=1.3in,foot=0.3in]{geometry}
\usepackage[T1]{fontenc}

\numberwithin{equation}{section}
\theoremstyle{plain}                
\newtheorem{theorem}{Theorem}[section]
\newtheorem{lemma}[theorem]{Lemma}
\newtheorem{proposition}[theorem]{Proposition}
\newtheorem{corollary}[theorem]{Corollary}
\newtheorem{convention}{Convention}
\theoremstyle{definition}           
\newtheorem{definition}[theorem]{Definition}

\theoremstyle{remark}
\newtheorem{remark}[theorem]{Remark}

\newcommand{\tot}{\tfrac{1}{2}} 

\newcommand{\abs}[1]{\left| #1 \right|} 

\newcommand{\set}[1]{\{#1\}} 
\newcommand{\sets}[2]{\set{#1\,:\,#2}} 
\newcommand{\lset}[1]{\left\{#1\right\}} 
\newcommand{\lsets}[2]{\lset{#1\,:\,#2}} 

\newcommand{\inds}[1]{ {\mathbf 1}_{\set{#1}}} 
\newcommand{\seq}[1]{\set{#1_n}_{n\in\N}} 
\newcommand{\seqz}[1]{\set{#1_n}_{n\in\N_0}} 

\renewcommand{\fam}[2]{\{ #1 \}_{#2}}

\newcommand{\ft}[2]{#1\dots#2} 
\renewcommand{\ft}[2]{#1,\dots,#2}

\newcommand{\prf}[1]{ \{ #1 \}_{t\in [0,T]}}

\newcommand{\dd}{{\mathrm d}}
\newcommand{\RN}[2]{\frac{\dd #1}{\dd #2}}

\providecommand{\R}{} \renewcommand{\R}{{\mathbb R}}
\providecommand{\C}{} \renewcommand{\C}{{\mathbb C}}

\newcommand{\N}{{\mathbb N}}
\newcommand{\PP}{{\mathbb P}}

\newcommand{\EE}{{\mathbb E}}
\newcommand{\FF}{{\mathcal F}}

\newcommand{\EN}{{\mathcal E}}

\renewcommand{\SS}{{\mathcal S}}
\renewcommand{\AA}{{\mathcal A}}

\newcommand{\eps}{\varepsilon}
\newcommand{\ld}{\lambda}

\newcommand{\define}[1]{{\bf #1}}

\newcommand\ta{{\tilde{a}}}
\newcommand\ha{{\hat{a}}}

\newcommand\hb{{\hat{b}}}

\newcommand\hd{{\hat{d}}}

\newcommand\bg{{\mathbb g}}

\newcommand\bu{{\mathbb u}}
\newcommand\tu{{\tilde{u}}}
\newcommand\hu{{\hat{u}}}

\newcommand\tB{{\tilde{B}}}

\newcommand\sG{{\mathcal G}}

\newcommand\sR{{\mathcal R}}

\newcommand\sS{{\mathcal S}}

\newcommand{\Sl}{S^{(\ld)}}

\newcommand{\AAl}{\AA}

\newcommand{\normsymbol}[1]{|#1|}

\newcommand{\ndr}[1]{\normsymbol{#1}_{\alpha}}

\newcommand{\nz}[1]{\normsymbol{#1}_0}

\newcommand{\nb}[1]{\normsymbol{#1}_{(\beta)}}

\newcommand{\nzd}[1]{\normsymbol{#1}_{\alpha}}
\newcommand{\nzkd}[1]{\normsymbol{#1}_{k+\alpha}}
\newcommand{\nzdd}[1]{\normsymbol{#1}_{1+\alpha}}
\newcommand{\nzddd}[1]{\normsymbol{#1}_{2+\alpha}}

\newcommand{\nkdr}[1]{\normsymbol{#1}_{k+\alpha}}

\renewcommand{\hd}[1]{[#1]_{\alpha}}
\newcommand{\hkd}[1]{[#1]_{k+\alpha}}

\newcommand{\hddd}[1]{[#1]_{2+\alpha}}

\newcommand{\s}[1]{\lceil #1 \rceil}
\newcommand{\sirbu}{S\^{\i}rbu}
\renewcommand{\i}[1]{\lfloor #1 \rfloor}

\renewcommand{\C}{C([0,T]\times \R)}
\newcommand{\Cb}{C_b([0,T]\times \R)}

\newcommand{\Cd}{C^{\alpha}([0,T]\times \R)}

\newcommand{\Ckd}{C^{k+\alpha}([0,T]\times \R)}
\newcommand{\Cdd}{C^{1+\alpha}([0,T]\times \R)} 
\newcommand{\Cddd}{C^{2+\alpha}([0,T]\times \R)} 
\newcommand{\Cds}{C^{\alpha}([0,T]\times \R\times\set{0,1})} 
 
\newcommand{\Cdds}{C^{1+\alpha}([0,T]\times \R\times\set{0,1})} 
\newcommand{\Cs}{C([0,T]\times \R\times\set{0,1})}

\newcommand{\Cddds}{C^{2+\alpha}([0,T]\times \R\times\set{0,1})}

\newcommand{\Ckdr}{C^{k+\alpha}(\R)}

\newcommand{\Cdddr}{C^{2+\alpha}(\R)}

\newcommand{\Cdddrs}{C^{2+\alpha}(\R\times\set{0,1})}

\newcommand{\bld}{\bar{\ld}}
\renewcommand{\bg}{\bar{g}}
\newcommand{\ug}{\underline{g}}

\newcommand{\hdel}{h^{(\delta)}}
\newcommand{\adel}{a^{(\delta)}}
\newcommand{\bdel}{b^{(\delta)}}

\newcommand{\bkey}{b^{(k)}}
\newcommand{\ukey}{u^{(k)}}
\newcommand{\bukey}{\bu^{(k)}}
\newcommand{\uukey}{\uu^{(k)}}
\newcommand{\hkey}{h^{(k)}}
\newcommand{\akey}{a^{(k)}}
\newcommand{\ldkey}{\ld^{(k)}}
\newcommand{\bldkey}{\bld^{(k)}}

\newcommand{\Eg}{E_{\gamma}}

\newcommand{\Dela}{\Delta_{\alpha}}

\newcommand{\peq}{\preceq}
\newcommand{\La}{L_{\alpha}}

\newcommand{\Ma}{M_{\alpha}}

\newcommand{\Lbd}{\Lambda^{\mathrm{bd}}}
\newcommand{\Sbd}{\sS^{\mathrm{bd}}}

\newcommand{\uld}{\underline{\ld}}
\newcommand{\uldkey}{\uld^{(k)}}

\newcommand{\uu}{\underline{u}}

\newcommand{\pil}{\pi^{(\ld)}}
\newcommand{\ul}{u^{(\ld)}}
\newcommand{\uli}{u^{(\ld),i}}
\newcommand{\uzi}{u^{(0),i}}
\newcommand{\ulone}{u^{(\ld_1)}}
\newcommand{\ultwo}{u^{(\ld_2)}}
\newcommand{\pilone}{\pi^{(\ld_1)}}
\newcommand{\piltwo}{\pi^{(\ld_2)}}

\newcommand{\utwo}{u^{(2)}}
\newcommand{\butwo}{\bu^{(2)}}
\newcommand{\uutwo}{\uu^{(2)}}
\newcommand{\htwo}{h^{(2)}}

\newcommand{\ldtwo}{\ld^{(2)}}
\newcommand{\bldtwo}{\bld^{(2)}}
\newcommand{\uldtwo}{\uld^{(2)}}

\newcommand{\udel}{u^{(\delta)}}
\newcommand{\budel}{\bu^{(\delta)}}
\newcommand{\uudel}{\uu^{(\delta)}}

\newcommand{\uone}{u^{(1)}}
\newcommand{\uuone}{\uu^{(1)}}
\newcommand{\buone}{\bu^{(1)}}

\newcommand{\bone}{b^{(1)}}
\newcommand{\ldone}{\ld^{(1)}}
\newcommand{\bldone}{\bld^{(1)}}
\newcommand{\uldone}{\uld^{(1)}}
\renewcommand{\bu}{\bar{u}}

\newcommand{\piil}{\pi^{(\ld),i}}

\begin{document}
 
  \begin{center}
    {\Large \textbf{\textsc{An example of a stochastic equilibrium
          with incomplete markets}}}

\vspace{0.3in}

\normalsize {\bf Gordan \v Zitkovi\' c}\footnote{ {\bf
    Acknowledgments.}  The author was supported in part by the
  National Science Foundation under award number DMS-0706947 during
  the preparation of this work. Any opinions, findings and conclusions
  or recommendations expressed in this material are those of the
  author and do not necessarily reflect those of the
  National Science Foundation.\\
  The author would like to thank Milica \v Cudina, Giovanni Leoni,
  Mihai \sirbu, Thaleia Zariphopoulou and  anonymous referees as well as
  the participants of the Workshop on Optimization in Mathematical
  Finance - Radon Institute, Linz, Austria, November 2008 and AMS
  Joint Meeting, Washington, January 2009 for valuable discussions. \\
  {\it Last updated: \today}
}\\[0.5ex]  Department of Mathematics\\
University of Texas at Austin\\
1 University Station, C1200\\
Austin, TX, USA\\[0.4ex]
{\tt gordanz@math.utexas.edu}\\
{\tt www.ma.utexas.edu/$\sim$gordanz}
  \end{center} 

\vspace{0.2in}

\begin{quote}
  \textbf{Abstract.} We prove existence and uniqueness of stochastic
  equilibria in a  class of incomplete continuous-time
  financial environments where the market participants are exponential
  utility maximizers with heterogeneous risk-aversion coefficients and
  general Markovian random endowments. The incompleteness featured in our
  setting - the source of which can be thought of as a credit event or
  a catastrophe - is genuine in the sense that not only the prices,
  but also the family of replicable claims itself is determined as a
  part of the equilibrium. Consequently, equilibrium allocations are
  not necessarily Pareto optimal and the related representative-agent
  techniques  cannot be used. Instead,
  we follow a novel route based on new stability results for a class
  of semilinear partial differential equations related to the
  Hamilton-Jacobi-Bellman equation for the agents'
  utility-maximization problems. This approach leads to a
  reformulation of the problem where the Banach fixed point theorem
  can be used not only to show existence and uniqueness, but also to
  provide a simple and efficient numerical procedure for its
  computation.
  \end{quote}

\vspace{0.1in}

\section{Introduction}
\label{intro}
\paragraph{\bf Market incompleteness and equilibria.}

The central theme of this paper is a study of an equilibrium problem
in an incomplete continuous-time stochastic setting. In contrast with
the complete case where significant advances have been made in
continuous time (see, e.g., \cite{AndRai08, DanPon92, Duf86, DufHua85,
  DufZam89,KarLakLehShr91, KarLehShr90, KarLehShr91, Zit06} as well as
Chapter 4.~of \cite{KarShr98}) the incomplete-market literature is
lagging behind. Even among the few incomplete markets studied so far (see
\cite{BasCuo98} or \cite{Hug08}, for example) ideas related to market
completeness, typically through the representative-agent approach, are
used. To the best of our knowledge, the present paper is the only one
in continuous time where a fully-incomplete market structure, in the
sense that both the prices and the set of replicable claims (the
marketed subspace) are determined as a part of the equilibrium, is
analyzed and existence of equilibria is established.

The difficulty with incomplete markets is that Pareto optimality,
which is commonly exploited to establish existence of equilibrium in a
complete market, is not guaranteed anymore. Our approach rests upon
the notion of stability of demand - a continuity property of the
optimal response (optimal portfolio) when viewed as a function of
market dynamics (the market-price-of-risk process). The leitmotif
behind such an analysis is the following: if the aggregate demand can
be shown to possess good continuity properties in an appropriate
topological setting, a fixed-point-type theorem can be used to
guarantee existence and (with some luck) uniqueness of an equilibrium
market dynamics.

\bigskip

\paragraph{\bf Stability in H\" older spaces}
The technical bulk of the present manuscript is devoted to the
stability of the optimal investment strategy under small functional
perturbations of the market-price-of-risk coefficient. Problems similar
to ours have attracted some attention recently - see, for example,
\cite{CarRas06,CarRas07,JouNap04,KarZit07,Lar07,LarZit06a}. There
are at least three - at first glance very different - reasons why such
problems are important.

First, from the statistical and financial point of view, it
is important to understand how misspecification or misestimation of
the market dynamics coefficients affects the optimal trading strategies
of agents who take the coefficient estimates at face value.
Equivalently, one can wonder what kinds of statistical procedures for
the estimation of those coefficients yield the most stability in
implementation. For a deeper discussion of this point, we refer the
reader to \cite{LarZit06a}. 

Second, in agreement with the classical methodology of the theory of
partial differential equations, and applied mathematics in general,
the following three aspects of every new problem are typically
studied: existence, uniqueness and sensitivity of the solution with
respect to changes of the problem's input parameters.  These criteria
are generally known as {\em Hadamard's well-posedness requirements}
(see \cite{Had02}). We view model specification as one of the most
important input data in the utility-maximization problem, and
understand the stability with respect to it as one of Hadamard's
requirements.

The last reason - by far the most important one for the present paper
- is to shed more light on the intimate relationship between the
notions of stability and that of competitive equilibrium in incomplete
markets. Indeed, rationality in the presence of individual preferences
can be viewed as the ``law of motion'' of financial agents. In
aggregate over several investors, the optimal trading strategies can
be interpreted as the agregate demand for the assets in question.
This aggregate demand, via the principles of market clearing,
ultimately determines the shape of the price-dynamics of the financial
assets.

Unlike other stability-related results mentioned above, the present
paper approaches the problem from a different point of view. While all
the previous strategies involved probabilistic and convex-analytic
techniques, we tackle the problem from a pure PDE perspective. In this
way we are able to draw much stronger and more precise conclusions
about the behavior of the function which maps the market-price-of-risk
process into the optimal portfolio.  Our main stability result is that
this map is {\em locally Lipschitz continuous} when the inputs and the
output are placed in anisotropic H\" older spaces on $[0,T]\times \R$,
and we give explicit estimates on the Lipschitz constant (see Theorem
\ref{thm:main}).  The price we pay for the PDE approach is not as dire
as one may think when placed in the equilibrium setting - the
necessary Markovian assumption only restricts the form of the agents'
random endowment to functions of state variables. The Markovian
structure of the resulting equilibrium asset dynamics is, on the other
hand, not a restriction. It can be viewed as a strenghtened version of
an equilibrium existence result: not only do we show that an
equilibrium exists, we also show that a {\em Markovian} equilibrium
exists.

In order to illustrate our techniques with the minimal amount of
distraction, we choose what can be termed as the simplest nontrivial
incomplete market model. Indeed, our market consists of a single risky
asset (and a unit riskless asset) whose dynamics depend on a single
Brownian motion and a one-jump Poisson process. All the incompleteness
in the market comes from this uninsurable and unpredictable jump.
Every agent is an exponential utility maximizer who also receives a
random endowment at the end of the time horizon and this random
endowment is allowed to depend on all sources of uncertainty in the
market, including the indicator of the unpredictable jump.  As a form
of a normalization, we assume that the volatility of the risky asset
is constant, but the drift can depend on time and the current value of
both state processes - the Brownian motion and the indicator of the
unpredictable jump.  It is clear that our setting can be generalized
in many different directions and that our estimates are far from
optimal. Our goal was to provide a proof of concept for the powerful
PDE techniques outlined in the body of the paper. We strive to keep
the presentation as simple as possible and as accessible as possible
to a reader who is not a specialist in PDE. For that reason, we use
only elementary techniques and results in the Schauder-type theory of
semilinear parabolic PDE as outlined, for example, in \cite{Kry96}. In
the same spirit, we do not pursue any connections with BSDE.

\bigskip

\paragraph{\bf Existence and uniquencess of the equilibrium price dynamics.}
Once the local Lipschitz continuity of the agents' demand functions is
established, we turn to the market-clearing conditions and show that
they can be rephrased in terms of a simple fixed-point problem for a
continuous map on an anisotropic H\" older space. Surprisingly, under
a ``smallness'' condition, a restriction of this map becomes a {\em
  contraction} and therefore, both existence and uniqueness of
equilibrium can be guaranteed. Furthermore, an efficient and
easy-to-implement numerical technique - based on iteration - emerges
naturally. Results of this type are very rare in equilibrium theory;
equilibria are typically not unique and their multiplicity is often
one of the major objections to the use of equilibrium modelling in
practice. Also, the computational methods (even in the
fully-finite-dimensional case) are often quite involved and based on
inefficient and relatively hard-to-implement procedures based on
Sperner's lemma and the related Scarf's algorithm (see \cite{Sca67}
and \cite{Sca67a}).

\bigskip

\paragraph{\bf The structure of the paper.} 
 Section 2.~desribes the model and provides
statements and some proofs of our main results. Section 3.~is devoted
to the proof of the central stability results, while Appendix A
contains the pertinent information about H\" older spaces.

\section{The financial environment and market clearing}
We define our financial environment
by specifying the structure of the three main ingredients: the
information structure, agents' preferences and the completeness
constraints.
\subsection{The information structure.}
\label{sse:info}
Let $T>0$ be a real number which we interpret as the  time
horizon, and let $\prf{B_t}$ be a standard Brownian motion defined on
a complete probability space $(\Omega,\FF,\PP)$. We assume that,
additionally, the same probability space accommodates an independent,
exponentially distributed random variable $\tau$ with parameter
$\mu>0$, where the corresponding counting process $\prf{N_t}$ is given
by
\[ N_t=\inds{\tau\leq t},\ t\in [0,T].\] Let $\prf{\FF_t}$ be the
right-continuous augmentation of the filtration $\prf{\FF^B_t\vee
  \FF^N_t}$, where $\FF^{B}_t=\sigma(\fam{B_s}{s\leq t})$ and
$\FF^N_t=\sigma(\fam{N_s}{s\leq t})$ for $t\in [0,T]$.

A more detailed interpretation of the roles of the two filtrations
will be given later. For now, let us just mention that the
$\sigma$-algebra $\FF_t$ models the total publicly-available
information at time $t$ and that $\tau$ models an event with a
significant impact on the market. Our model does not deal with
asymmetric-information situations - all the agents have access to
the same information, and no information is hidden from any of them. As
we shall explain below, it is the way in which the information
trickles into the dynamics of traded assets that makes the market
incomplete.

\subsection{Completeness constraints.} 
\label{sse:compl}
The final goal of our analysis is the determination of the form of an
equilibrium asset-price process. Without exogenously-imposed
constraints, there is nothing that will prevent the agents from
``opening'' as many markets as possible, and, eventually, building a
market structure that will be able to replicate any uncertain
pay-off. In other words, without constraints - and bar pathologies -
all equilibrium markets are necessarily complete. The situation we are
modelling, however, calls for a degree of incompleteness: we envision
the situation in which the random variable $\tau$ marks an event with
the property that the time-scale at which the market adjusts to $\tau$
is larger than the time-scale on which it reacts to the other
information.  Consequently, as we shall see, its nature will be such
that no perfect insurance against its effects can be bought or sold.

More generally, market environments where a portion of the information
flow affects prices faster than the rest are important examples of
natural constraints that preclude completeness; we term them
\define{fast-and-slow-information market environments}. In fact, many
widely used incomplete market models (virtually all models where the
incompleteness stems from more ``sources of uncertainty'' than
available risky assets) can be viewed as equilibria in
fast-and-slow-information market environments.  The situation
described in the present paper corresponds to, arguably, the simplest
continuous-time fast-and-slow-information market environment. More
complicated (and realistic) environments can be constructed and
analyzed, but we opt for one which allows us to showcase our methods
without unnecessary confusion.

The defining properties of $\tau$ lead to the 
following class of possible
asset-price dynamics: let $\Lambda$ be the family of all
$\prf{\FF_t}$-predictable processes $\prf{\ld_t}$ such that $\int_0^T
\abs{\ld_u}\, du<\infty$, a.s., and let
$\SS=\lsets{\prf{\Sl_t}}{\ld\in\Lambda}$ be the family of all It\^
o-processes $\prf{\Sl_t}$ with the dynamics given by 
\begin{equation}
  \label{equ:SS}
  \begin{split}
    d\Sl_t= \ld_t\, dt + dB_t, \ \Sl_0=0,
  \end{split}
\end{equation}
as $\ld$ ranges over $\Lambda$. Additionally, let $\Lbd$ be the subset
of $\Lambda$ consiting of all uniformly bounded $\ld$. The corresponding
set of price processes is denoted by $\Sbd$.  We always (implicitly) assume
that a risk-free num\' eraire asset $\bar{S}\equiv 1$ accompanies each
$\Sl$.

\begin{remark}\ \label{rem:div}
\begin{enumerate}
\item It is important to note that there is no loss of generality in
  assuming that the volatility coefficient is identically equal to $1$
  or that ``arithmetic'' dynamics for $\Sl$ are chosen instead the
  more common ``geometric'' one. Indeed, the only property of the
  asset-price dynamics important for the determination of the
  equilibrium is the set of replicable random variables it produces,
  and this depends on drift and volatility only through their quotient
  $\ld$ (the so-called market price of risk). In fact, our formulation
  of the equilibrium simply {\em cannot} distinguish between the
  two. This would be the case in the more complicated (and more
  realistic) model where the agents' random endowments depend on the
  asset-price itself, i.e., where the agents hold financial derivatives
  of the asset.  Similarly, the postulated availability of a
  trivial num\' eraire asset $\bar{S}$ involves a minimal loss of
  generality. Indeed, the agents' utilities depend solely on the terminal
  wealth and do not exhibit any time-impatience characteristics. In
  fact, without introducing consumption into the model, it is   
  impossible to disentangle the two solely on the basis of equilibrium
  analysis.
\item The two terms on the right-hand side of (\ref{equ:SS})
  contribute to the fluctuations in $\Sl$ on different scales; the
  order of magnitude of the term $\ld\, dt$ is $dt$ while the order of
  (the absolute value of) the term $dB_t$ is
  $\sqrt{dt}$. Equivalently, the first term takes on average
  $1/\sqrt{dt}$ times longer to produce the same (local) efect as the
  second one. In this sense, the fact that the new information
  regarding $\tau$ comes in only through the first, slower, term
  corresponds exactly to the requirement that the occurrence of $\tau$
  be absorbed into the asset dynamics on a slower scale than the
  shocks produced by $B$.
\item One can view \eqref{equ:SS} as a restriction on the set of
  stochastic processes that may serve as allowable market
  dynamics. For a different specification of the characteristics of the
  market environment (constraints on the information flow, number of
  assets, etc), we would get different families of
  processes. Therefore, one could give the following, abstract,
  definition: a \define{completeness constraint} is simply a family
  $\SS$ of (possibly multi-dimensional) semimartingales. 
\item While no dividend structure has been mentioned so far, there is
  a way of incorporating dividends in the present model, and, perhaps,
  the best way to describe it is to compare it to a particular
  incomplete setting - the one with short-lived securities - in the
  language of \cite{MagQui96} (see also how it leads to
  existence results in discrete time as in \cite{MagQui96},
  Proposition 25.1, p. 255). One can think of the asset $\Sl$ as an
  aggregation of short-lived assets with inception $t$ and maturity
  $t+dt$, each of which pays a dividend $\Sl_{t+dt}$ and costs
  $\Sl_t$. Equivalently, the agent enters into a bet a unit of which
  pays $dB_t$ and costs $-\ld_t\, dt$. In this way, we can think of
  $dB_t$ as the local dividend and $-\ld_t$ as its price. It should be
  noted that it is easy to incorporate any (reasonably behaved)
  martingale $M_t$ of the form $M_t=\int_0^t \sigma_u\, dB_u$ as the
  divided process - one has to replace $\Sl$ by $\int_0^{\cdot}
  \sigma_u\, d\Sl_u$. Even if $M$ is not a martingale, we can simply
  add its drift to $-\ld$ and turn it into one. This procedure also
  allows us to give a loose interpretation of the formation of
  volatility in the market. While the market price of risk $\ld$ is
  determined in the equilibrium from the agent's primitives and the
  structural properties of the dividends, the volatility is determined
  by their quantitative properties. Consequently, the unit-volatility
  assumption we impose on $\Sl$ can be reinterpreted as a simple
  normalization of the dividends.
\end{enumerate}
\end{remark}
\subsection{The agents and their preferences.} 
\label{sse:agents}
We assume there is a finite number $I\in\N$ of agents, all of whom
actively participate in trading in all available assets. 
The preference structure of each one of them is determined by the
following:
\begin{enumerate}
\item[(I)] the {\em utility functions:} each agent is an exponential-utility maximizer with \[U^i(x)=-\exp(-\gamma^i x), x\in\R,\ \gamma^i>0.\]
\item[(II)] the {\em random endowment:} $\EN^i=g^i(B_T,N_T)$, with bounded $g^i:\R\times \set{0,1}\to \R$, where further regularity conditions  are to be specified. 

\end{enumerate}
As usual in the Alt-von Neumann-Morgenstern expected utility paradigm,
agent $i$ prefers the random variable $X_1\in\FF_T$ to a random
variable $X_2\in\FF_T$ if and only if $\EE [ U^i( X_1+\EN^i) ] \geq
\EE[ U^i( X_2+\EN^i)]$, where we set $\EE[Y]=-\infty$ whenever
$\EE[Y^+]=\EE[Y^-]=+\infty$.  
\begin{remark}\ 
\begin{enumerate}
\item 
  If one thinks about $B$ and $N$ as factors, (II) above states
  that the agents' random endowment depends on all factors driving the
  public information. In particular, the (conditional) distribution of
  the random endowment may change abruptly and considerably at time
  $\tau$. One of the possible financial interpretations of the
  situation is that all agents hold (long or short) positions in
  assets whose pay-offs are affected by the occurrence of $\tau$
  (default-sensitive derivatives, callable bonds, disaster-sensitive
  investments, etc.).
\item The exponential nature of the utilities allows us to partially
  remove the assumption that all agents use the same (subjective)
  probability to compute the expected utility of a particular
  position. Indeed, using the identity
  \[ \EE^{\PP^i}[ -\exp(-\gamma^i (X+\EN^i)) ] = \EE[ -\exp(-\gamma^i
  (X+\tilde{\EN}^i))],\] where $\tilde{\EN}^i=
  \EN^i-\tfrac{1}{\gamma^i}\log( \RN{\PP^i}{\PP} )$, we can easily
  ``absorb'' different subjective probabilities into the form of the
  random endowment. Care must be taken, though, to ensure that the
  appropriate integrability conditions are met. Also, it should be
  noted that such a change of measure can lead to loss of the
  Markovian structure of the ingredients.
\end{enumerate}
\end{remark}
Let us now focus on the pertinent case when the set of tradable assets
consists of a single risky asset given by $\Sl\in\Sbd$, for some
bounded market-price-of-risk process $\ld\in\Lbd$ and the trivial
num\' eraire asset $\bar{S}$ (as described just below equation
\eqref{equ:SS}, above).  Agent $i$ chooses a dynamic
self-financing portfolio process in the appropriate admissibility
class (to be specified shortly) so as to  maximize the expected
utility of terminal wealth:
\begin{equation}%
  \label{equ:vi-ld}
  \begin{split}
    \EE[ U^i(\textstyle\int_0^T \pi_u\,
    d\Sl_u+\EN^i)]\to\max.
  \end{split}
\end{equation}
Here, the value $\pi_t$ of the one-dimensional process $\prf{\pi_t}$
denotes the number of shares of the risky asset in the portfolio. We
do not explicitly mention the number of shares $\rho_t$ of the
riskless asset, since, thanks to the self-financing condition, it
necessarily satisfies
\begin{equation}
   \label{equ:self-fin}
   \begin{split}
\rho_t=\int_0^t \pi_u\, d\Sl_u -\pi_t
\Sl_t.
   \end{split}
\end{equation}

Our class of admissible portfolio processes aims to be just
restrictive enough to rule out doubling strategies, yet large enough
to contain the maximizer of the utility-maximization problem
\eqref{equ:vi-ld}. We would like to emphasize that, due to the
regularity of some of the ingredients, one does not need the
sophistication typically encountered in general semimartingale models
(see, e.g., the classes $\Theta_i$, $i=1,2,3,4$ in
\cite{DelGraRheSamSchStr02} or the notion of permissibility in
\cite{OweZit06}). Instead, we sacrifice a small amount of generality
for a large gain in simplicity by proceeding as follows: the
admissibility class $\AAl$ consists of all $\prf{\FF_t}$-predictable
processes $\prf{\pi_t}$ such that $\EE[\exp(\int_0^T
  (\tot+\eps)\pi_u^2\, du)]<\infty$, for some $\eps>0$.  Note that, by
H\" older's inequality, the integral $\int_0^T \pi_u\, d\Sl_u$ is
well-defined for each $\pi\in\AAl$ and $\ld\in\Lambda$.
\begin{remark}
  As we shall see later,  the optimization problem
  \eqref{equ:vi-ld} admits a $dt\times d\PP$-unique maximizer $\piil$
  of \eqref{equ:vi-ld} in the class $\AAl$ for each $\ld\in\Lbd$ and each agent $i=\ft{1}{I}$.  The
  portfolio process $\prf{\piil_t}$ is interpreted as an
  {\em optimal response} of agent $i$ to the market dynamics
  induced by $\ld$; the mapping $\prf{\lambda_t} \mapsto
  \prf{\piil_t}$ plays a central role in our analysis.
\end{remark}
\subsection{Market-clearing conditions and the main result}
Having introduced all the ingredients, we turn to our main problem. A
central economic principle is that the prevailing market dynamics
must have the following fundamental property: the demand and supply
for each tradable asset must match, i.e., all markets must clear.  A
precise definition follows:
\begin{definition}
\label{def:equilibrium}
  A process $\Sl\in\SS$ (or, equivalently, its market-price-of-risk
  process $\ld\in\Lambda$) is said to be an \define{equilibrium price
    dynamics (market price of risk)} if there exist processes
  $\prf{\piil_t}\in\AAl$, $i=\ft{1}{I}$ such that
  \begin{enumerate}
  \item {\bf (rationality)} for all $\pi\in\AAl$ and all for all
    $i=\ft{1}{I}$,
\[\EE[U^i(\textstyle\int_0^T \piil_u\, d\Sl_u+\EN^i)]\geq
    \EE[U^i(\int_0^T \pi_u\, d\Sl_u+\EN^i)],\text{ and }\] 
  \item {\bf (market clearing)} $\sum_{i=1}^I \piil_t=0$, for all $t\in
    [0,T]$, a.s.
  \end{enumerate}
\end{definition}
\begin{remark}
  The reader will immediately realize that we only require market
  clearing for the risky asset $\Sl$. The self-financing condition
  \eqref{equ:self-fin} implies, however, that the market for the
  riskless asset $\bar{S}$ will clear, in that case, as well.
\end{remark}

The main result of the present paper - which asserts the existence,
uniqueness and efficient computability of equilibrium price dynamics -
is summarized in Theorem \ref{thm:main} below (we direct the reader to
Appendix \ref{sec:appendix} for details on H\" older spaces and the related
notation):
\begin{theorem}[Existence, Uniqueness and Computability of Equilibria]
\label{thm:main}
Consider the setup given in subsections \ref{sse:info},
\ref{sse:compl} and \ref{sse:agents}, and assume, additionally, that
$g^i(\cdot,n)\in \Cdddr$, for $i=\ft{1}{I}$, $n=0,1$ and some
$\alpha\in (0,1]$.  Then there exists a constant $T_0>0$ (which may
depend on $\alpha, \mu, \{\gamma^i\}_{i=\ft{1}{I}}$ and $\{
\nzddd{g^i}\}_{i=\ft{1}{I}}$) such that for $T\leq T_0$,
\begin{enumerate}
\item the equilibrium price dynamics \emph{exist} and the corresponding
  market-price-of-risk process is \emph{unique} in the class of all processes
  $\prf{\ld_t}$ such that $\ld_t=\ld(t,B_t,N_{t-})$, $t\in
  [0,T]$, for some function $\ld\in \Cdds$.
\item the function $\ld$ is \emph{efficiently computable}, i.e., there
  exists a subset $D\subseteq \Cdds$ with $0\in D$ and a contraction
  $\Pi:D\to D$ such that $\ld$ is the fixed point of $\Pi$.
\end{enumerate}
\end{theorem}
\begin{remark}\ 
\begin{enumerate}
\item   The most restrictive feature of Theorem \ref{thm:main} is the
  ``smallness'' condition we place on the size $T$ of the time
  horizon. It is not evident that it can be dealt with by ``pasting''
  equilibria together as the upper bound on $T$ depends on the
  primitives of the model. In particular, extending an equilibrium
  from $T$ to $T'>T$ would entail using the value of the original
  equilibrium at time $T$ as the terminal condition of the equilibrium
  on $[T,T']$. An iteration of this procedure may very well lead to a bounded
  sequence. 

  We do conjecture, however, that this assumption {\em can} be relaxed
  if one only wants to establish the existence of the equilibrium, but
  {\em cannot} be relaxed if the additional benefits of uniqueness and
  efficient computability are also desired. The resolution of those
  important questions is the content of our future work and seem to
  require quite different and more sophisticated techniques. Moreover,
  preliminary numerical experiments suggest that the for a large class
  of realistic parameter values, the smallness constraint does not
  seem to be binding and the iteration procedure converges.
\item It is, perhaps, instructive to interpret the above uniqueness
  result in the light of part (4) of Remark \ref{rem:div}.  What is
  truly unique is the equilibrium price of the local gamble of size
  $dW_t$. It follows that each {\em replicable} derivative security
  $C$ admits a unique equilibrium price process $S$. The difference
  with the complete models is that the class of replicable contingent
  claims is determined as a part of the equilibrium and admits no
  simple description a priori.
\end{enumerate}
\end{remark}
The proof of Theorem \ref{thm:main} is based on the following
stability estimate which, as discussed in the introduction, is of
interest in its own right. We devote the entire section
\ref{sec:stability} to its proof. Since it deals with a single-agent's
optimization problem, we omit the agent index $i$ from its
statement. Also, in the interest of readability, we introduce some
additional notation: for a function $u: [0,T]\times \R\times
\set{0,1}\to \R$, we define the \define{$n$-difference} $u_n$ as
\[ u_n(t,x,n)= u(t,x,1)-u(t,x,n)=
\begin{cases}
  u(t,x,1)-u(t,x,0), & n=0,\\
  0 , & n=1.
\end{cases}
\]
For such functions, restrictions $u(\cdot,\cdot,0)$ and
$u(\cdot,\cdot,1)$ are called the $n=0$- and the $n=1$-\define{slices}
of $u$. We say that a measurable function $f:[0,T]\times \R\times
\set{0,1}\to \R$ is a \define{Markov re\-pre\-sen\-ta\-tive} of the
process $\prf{Y_t}$ if
\[ Y_t=f(t,B_t,N_{t-}),\text{for all $t\in [0,T]$,\text{ a.s.}}\] It is
convenient to identify a process admitting a Markov representative
with the
Markov representative itself. In fact, we will do so from this point on
with little or no explicit mention.
\begin{theorem}[Stability of the Optimal Portfolio]
\label{thm:stability}  
  For $\ld\in\Cds$ and $g\in\Cdddrs$ there exists a unique function
  $\ul\in\Cddds$ which solves the following Cauchy problem for a
  semilinear PDDE (partial diferential-difference equation):
  \begin{equation}
    \label{equ:PDDE}
    \left\{
      \begin{split}
        & 0=\ul_t + \tot \ul_{xx} - \ld \ul_x + \tfrac{1}{2\gamma} \ld^2
        - \tfrac{\mu}{\gamma}(\exp(-\gamma \ul_n)-1),\text{ on } [0,T)\times \R\times\set{0,1},\\
        & \ul(T,\cdot,\cdot) =g.
      \end{split}
    \right.
  \end{equation}
  The stochastic process $\prf{\pil_t}$ with the Markov representative
  $\pil=\tfrac{1}{\gamma}\ld -\ul_x$, i.e., the process given by
  \begin{equation}
    \label{equ:form-of-pi}
    \begin{split}
      \pil_t= \tfrac{1}{\gamma} \ld(t,B_t,N_{t-})-\ul_x(t,B_t,N_{t-}),\ t\in
      [0,T],
    \end{split}
  \end{equation}
  is in $\AAl$  and attains the maximum in the utility-maximization
  problem \eqref{equ:vi-ld}. 

  The maps $\ld\mapsto \pil$, $\ld\mapsto \ul_x$, from $\Cds$ to
  itself, are locally Lipschitz continuous.  More precisely, there
  exists a constant $C=C(\gamma,\mu,\alpha)>0$ such that for each
  $R>0$ and $\ld_1, \ld_2 \in\Cds$ with $\nzd{\ld_1},\nzd{\ld_2}\leq
  R$ we have
\begin{equation}
   \label{equ:lip-ulx}
   \begin{split}
     \nzd{\ulone_x-\ultwo_x}\leq L(R) \nzd{\ld_1-\ld_2}, \text{ and
     }\nzd{\pilone-\piltwo}\leq \left(\tfrac{1}{\gamma}+L(R)\right)
     \nzd{\ld_1-\ld_2},
   \end{split}
\end{equation}
where
\begin{equation}
   \label{equ:def-L(R)}
   \begin{split}
     L(R)= C\,T^{(1+\alpha)/(2+\alpha)} e^{e^{2+ 2\gamma \nz{g}+ T R^2+2 \mu
         T}}\Big(\nzddd{g}+(1+T)(1+R^2)\Big)^{6+4\alpha}.
   \end{split}
\end{equation}
\end{theorem}

With Theorem \ref{thm:stability} at our disposal, the proof of Theorem
\ref{thm:main} becomes straightforward. We start with a
characterization which follows directly from Definition
\ref{def:equilibrium} and equation (\ref{equ:form-of-pi}) of Theorem
\ref{thm:stability}. When the index $i$ is added to $\ul$ (or its
derivatives), as in $\uli$ (or $\uli_x$, etc.), we are referring to
the solution to (\ref{equ:PDDE}) with the agent-dependent terminal
condition $g=g^i$ and the risk-aversion parameter $\gamma=\gamma^i$.
\begin{lemma}
\label{lem:eqil-char}
A process $\prf{\ld_t}$ with Markov
representative $\ld\in\Cds$ is an equilibrium market-price-of-risk if
and only if it is a fixed point of the operator $\Pi:\Cds\to\Cds$
defined by
\[ \Pi (\ld)= \bar{\gamma} \sum_{i=1}^I \uli_x,\]
where $\bar{\gamma}=\Big(\sum_{i=1}^I \tfrac{1}{\gamma^i}\Big)^{-1}$.
\end{lemma}
In the sequel, let $R_0=\tfrac{2}{\bar{\gamma}}\sum_{i=1}^I \nzd{\uzi_x}$ (where
$\uzi_x$ corresponds to $\ld\equiv 0$)  be
computed for $T=1$. Note that $R_0$ dominates the value of the same
expression we get 
when we set $T<1$. 
\begin{lemma}
\label{lem:into}
  There exists a constant $T_1=T_1(\alpha, \mu,
  \{\gamma^i\}_{i=\ft{1}{I}},\{ \nzddd{g^i}\}_{i=\ft{1}{I}})>0$, such that if
  $T\leq T_1$, $\Pi$ maps the ball
  \[ B_{\alpha}(R_0)=\sets{\ld\in\Cds}{\nzd{\ld}\leq  R_0},\] of
  radius $R_0$ in $\Cds$ into itself.
\end{lemma}
\begin{proof}
  The first inequality in (\ref{equ:lip-ulx}) of Theorem \ref{thm:stability} implies
  that there exists a non-decreasing function $F:[0,\infty)\to
  [0,\infty)$ such that if
  $T\leq 1$ we have
  \[ \nzd{\uli_x}\leq \nzd{\uzi_x}+T^{1/(2+\alpha)}
  F(\nzd{\ld}),\text{ and so } \nzd{\Pi(\ld)}\leq \tot R_0 +
  \tfrac{1}{\bar{\gamma}} I\, T^{1/(2+\alpha)}F(\nzd{\ld}), \] for
  $\ld\in\Cds$ and all $i=\ft{1}{I}$. Note that $F$, being derived from $L(\cdot)$ of
  \eqref{equ:def-L(R)}, depends on $\alpha$, $\mu$, $\gamma^i$,
  $\nzddd{g^i}$, $i=\ft{1}{I}$ but not on $T$ or $\ld$ for $T \le 1$. 
 Therefore, for $T\leq T_1=1\wedge
    \Big(\tfrac{\bar{\gamma} R_0}{2 I F(R_0)}\Big)^{2+\alpha}$,
  we  have $ \nzd{\Pi(\ld)} \leq R_0$, whenever $\nzd{\ld}\leq
  R_0$.
\end{proof}
\begin{lemma}
\label{lem:contr-into}
  There exists a constant $T_0=T_0(\alpha, \mu,
  \{\gamma^i\}_{i=\ft{1}{I}},\{ \nzddd{g^i}\}_{i=\ft{1}{I}})>0$ such that if
  $T\leq T_0$, the mapping $\Pi$ is a contraction from $
  B_{\alpha}( R_0)$ into itself. 
\end{lemma}
\begin{proof}
  By Lemma \ref{lem:into}, the mapping $\Pi$ maps $B_{\alpha}(R_0)$
  into itself, as long as $T\leq T_1$. Therefore, we can use the
  Lipschitz estimate \eqref{equ:lip-ulx} of Theorem \ref{thm:stability}
  and reasoning similar to the one in the proof of Lemma
  \ref{lem:into} to conclude that there exists  $T_0 \le T_1$ such
  that for $T\leq T_0$,
  we have
\[ \nzd{\Pi(\ld_1)-\Pi(\ld_2)} \leq \tot \nzd{\ld_1-\ld_2}, \]
for $\ld_1,\ld_2\in B_{\alpha}(R_0)$. 
\end{proof}
\begin{proof}[Proof of Theorem \ref{thm:main}] The path to the proof
  is paved by Lemmas \ref{lem:eqil-char}, \ref{lem:into}
  and \ref{lem:contr-into}. Indeed, we simply apply the Banach fixed
  point theorem to the mapping $\Pi$ on the complete metric space
  $B_{\alpha}(R_0)$ and note that, by the first part of Theorem
  \ref{thm:stability}, 
  $\Pi(\ld)\in\Cdds$, as soon as $\ld\in\Cds$.
\end{proof}

\section{Stability of the optimal portfolio: a proof of Theorem 
\ref{thm:stability}}
\label{sec:stability}
The purpose of the present section is to prove Theorem
\ref{thm:stability}. We break its statement into several smaller, more
manageable results, and proceed to discuss them one by one. 

\subsection{The HJB equation: existence and verification}
The first assertion of Theorem \ref{thm:stability} is that a the
PDDE in (\ref{equ:PDDE})  admits a regular
solution and that it can be used to construct a Markov representative
for the optimal portfolio in the utility maximization problem (\ref{equ:vi-ld}).  Similar
characterizations are ubiquitous in stochastic control in the case
of exponential utility, and are too numerous to list; we simply instruct
the reader to consult \cite{Car08} and references therein. We note
that a PDDE is a special case of a partial integro-differential
equation, but we choose not to use that term because of the simplicity
of the integral component. Equivalently, it could have been stated as
a system of one linear and one semilinear PDE. 

Appendix \ref{sec:appendix} should be consulted for
notation not explicitly introduced in the main body of the
paper. However, for the reader's convenience, we state the following convention
both here and in subsection \ref{sse:convention} of Appendix \ref{sec:appendix}:
\begin{convention}\ 
\label{con:constantsmain}
\begin{enumerate}
\item The variables $\mu$, $\gamma$ and $\alpha\in (0,1)$ are
  considered ``global'' and will not change throughout the paper. Any
  function of the global variables (and global variables only) if
  called a \define{universal} constant.
\item The notation $a\preceq b$ means that there exists
  a universal constant $C>0$ such that $a\leq C b$.  Such a constant
  may change from line to line.
\end{enumerate}
\end{convention}

Our analysis starts with a simple linear existence result in the
spirit of Schauder's theory:
\begin{lemma}
  \label{lem:lin-prob}
For $h,a\in\Cd$ and $g\in\Cdddr$ the Cauchy problem
  \begin{equation}
    \label{equ:lin-prob}
    \left\{   \begin{split}
        &0=u_t+\tot u_{xx}+ h u_x+a \text{ on } [0,T)\times \R, \\
        &u(T,\cdot)=g(\cdot)
      \end{split}
    \right.
  \end{equation}
  admits a unique solution $u\in\Cddd$. Moreover,
for all $\beta>0$ we have
  \begin{eqnarray}
    \label{equ:beta}
    \nb{u}&\leq& \tfrac{1}{\beta} \nb{a}+\nz{g}.
  \end{eqnarray}
\end{lemma}
\begin{proof}
  The existence of a unique solution to \eqref{equ:lin-prob} in
  $\Cddd$ is well-known (see, for example \cite{Kry96}, Theorem 9.2.3,
  p. 140).  To get (\ref{equ:beta}), we use the fact that if $u$ is
  the unique solution to \eqref{equ:lin-prob}, then
  $\tu(t,x)=e^{-\beta(T-t)} u(t,x)$ solves
  \begin{equation}
    \label{equ:lin-prob-beta}
    \left\{   \begin{split}
        &0=\tu_t+\tot \tu_{xx}- \beta \tu+h \tu_x+\ta 
        \text{ on } [0,T)\times \R, \\
        &\tu(T,\cdot)=g(\cdot),
      \end{split}
    \right.
  \end{equation}
  where $\ta(t,x)=e^{-\beta(T-t)} a(t,x)$. It remains to note that
  the constant function $w(t,x)=\nz{g}+\tfrac{1}{\beta}
  \nz{\ta}=\nz{g}+\tfrac{1}{\beta} \nb{a}$ is a subsolution, and that
  its negative $-w$ is a supersolution of \eqref{equ:lin-prob-beta};
  the maximum principle (see \cite{Kry96}, p.105, Theorem 8.1.2) 
  implies 
  \begin{equation}\label{equ:max-prin}
    \begin{split}
      -w(t,x)\leq \tu(t,w)\leq w(t,x), \qquad (t,x)\in [0,T]\times\R. 
    \end{split}
\end{equation}
\end{proof}

When a non-linear term is added, a similar result can be obtained
with a bit more work. The argument is based in part on the following lemma:
\begin{lemma}
  \label{lem:sequence}
  Let $\seqz{x}$ be a sequence of nonnegative real numbers with the
  property that
  \begin{equation}
    \label{equ:recurr}
    \begin{split}
      x_{n+1}\leq A+B x_n^{\alpha},
    \end{split}
  \end{equation}
  for some constants $A,B\geq 0$ and $0<\alpha <1$.  Then there exists
  a constant $g(\alpha)>1$, independent of $A,B$ and $\seq{x}$, such
  that
  \begin{equation}
    \label{equ:bound-on-xn}
    \begin{split}
      \limsup_n x_n \leq g(\alpha) \max(A,B^{1/(1-\alpha)}).
    \end{split}
  \end{equation}
In particular, the sequence $\seq{x}$ is bounded. 
\end{lemma}
\begin{proof}
  If $B=0$, the inequality in \eqref{equ:bound-on-xn} clearly
  holds. When $A=0$ and $B>0$, we have $x_{n}\leq
  B^{1+\alpha+\alpha^2+\dots+\alpha^{n-1}} x^{\alpha^n}_0$, which
  implies \eqref{equ:bound-on-xn} directly.

  Focusing on the case $A,B>0$, we set $\tB=B/A^{1-\alpha}$, and let
  $G>1$ be the unique positive solution to $G=1+ \tB G^{\alpha}$.  The
  scaled sequence $y_n=x_n/(AG)$, $n\in\N$ satisfies
  \begin{equation}
    \begin{split}
      y_{n+1} &=\tfrac{1}{AG} x_n \leq 1/G+\tB
      G^{\alpha-1}(y_n)^{\alpha} =\kappa + (1-\kappa) y_n^{\alpha},
    \end{split}
  \end{equation}
  where $\kappa= G^{-1}= 1-\tB G^{\alpha-1}$.  Therefore, if
  $y_{n}\leq 1$, then $y_{n+k}\leq 1$, for all $k\in\N$. On the other
  hand, suppose that $y_n >1$, for all $n\in\N_0$. Then, for
  $n\in\N_0$, we have $y_{n+1} \leq y_n^{\alpha}$, and so, $y_n\leq
  y_{0}^{\alpha^{n}}$. Consequently, $\limsup_{n\to\infty}
  y_n\leq 1$, i.e.,
  $\limsup_n x_n \leq A G$.

  It remains to show that $AG$ is bounded from above by the expression
  on the right-hand side of \eqref{equ:bound-on-xn}. Let $g(\alpha)$
  be the unique positive solution of $g(\alpha)=1+g(\alpha)^{\alpha}$,
  so that $g(\alpha)>1$. If $\tB\geq 1$, then
  \[ 1+ \tB (g(\alpha)\tB^{1/(1-\alpha)})^{\alpha}= 1+(g(\alpha)-1)
  \tB^{1/(1-\alpha)} \leq g(\alpha)\tB^{1/(1-\alpha)}.\] Therefore,
  the monotonicity   of the function $x\mapsto
  x-1-\tB x^{\alpha}$ implies that $G\leq g(\alpha) \tB^{1/(1-\alpha)}$.
  When $ \tB\leq 1$, we have $1+\tB g(\alpha)^{\alpha} \leq
  g(\alpha)$, so $G\leq g(\alpha)$. Therefore,
  \[ A G\leq A g(\alpha) \max(1,\tB^{1/(1-\alpha)})= g(\alpha)\max(A,
  B^{1/(1-\alpha)}).\]
\end{proof}
The existence of a $C^2$ solution to semilinear equations of the type
treated in Proposition \ref{pro:semilinear} is well-known (see, for
example, \cite{Bec01a}). However, the $C^{2+\alpha}$-regularity
and the $\abs{ \cdot }_0$ and $\hddd{\cdot}$ estimates 
need additional work in this proof. The reader will note that we
still go through the steps of the existence argument; the reason is
that the H\" older regularity is established through the
$\nz{\cdot}$-convergent sequence which is originally constructed to
show the existence of a $C^2$ solution.
\begin{proposition}
  \label{pro:semilinear}
  Consider a semilinear Cauchy problem of the form
  \begin{equation}
    \label{equ:semi}
    \left\{   \begin{split}
        &0=u_t+\tot u_{xx}+ h u_x+a- b e^{\gamma u} 
           \text{ on } [0,T)\times \R, \\
        &u(T,\cdot)=g,
      \end{split}
    \right.
  \end{equation}
  where $g\in \Cdddr$, $h\in\Cd$, $a,b\in\Cd$, $\gamma>0$ and
  $a(t,x)\geq 0$, $b(t,x)\geq 0$, for all $(t,x)\in [0,T]\times \R$.
  Then (\ref{equ:semi}) admits a unique  solution in
  $\Cddd$.  Furthermore, 
the following bounds hold for all $(t,x)\in
  [0,T]\times \R$:
\begin{equation}
  \label{equ:bounds-nz-u}
   \begin{split}
-\tfrac{1}{\gamma}\log(e^{\gamma \nz{g}}+\gamma T \nz{b}) \leq
     u(t,x) \leq \nz{g}+T \nz{a},
   \end{split}
\end{equation}
 \begin{equation}
\label{equ:bound-ddd}
\begin{split}\hddd{u}&\peq \hddd{g}+ \hd{a} + 
  \hd{b}e^{\gamma (\nz{g}+T\nz{a})}+ \big(1+
  \nz{b}^{1+\alpha}e^{\gamma(1+\tot\alpha)(T\nz{a}+\nz{g})} +
  \nzd{h}^{2+\alpha}\big)\nz{u}.
    \end{split}
\end{equation}
\end{proposition}
\begin{proof}
  We start by setting $F(t,x,y)=a(t,x)-b(t,x) \exp(\gamma y)$, and
  note that, by Lemma \ref{lem:F-delta},
  $F(\cdot,\cdot,w(\cdot,\cdot))$ is in $\Cd$, whenever $w$ is.
  Therefore, by Lemma \ref{lem:lin-prob}, for each $w\in\Cd$, there
  exists a unique $\Cddd$-solution $u$ of the Cauchy problem
  \begin{equation}
    \label{equ:lin-prob-2}
    \left\{   \begin{split}
        &0=u_t+\tot u_{xx}+ 
        h u_x+F(\cdot,\cdot,w(\cdot,\cdot)) \text{ on } [0,T)\times \R, \\
        &u(T,\cdot)=g(\cdot).
      \end{split}
    \right.
  \end{equation}
Consequently, due to Lemma \ref{lem:lin-prob}, the operator $\sG:\Cd\to\Cddd$ which assigns to $w\in\Cd$ the 
unique solution  of \eqref{equ:lin-prob-2} is well-defined. 
For $u_1,u_2\in\Cd$, the function $\kappa=\sG u_1-\sG
  u_2$ satisfies
  \begin{equation}
    \label{equ:lin-prob-3}
    \left\{   \begin{split}
        &0=\kappa_t+\tot \kappa_{xx}+ h \kappa_x+\delta \text{ on } [0,T)\times \R, \\
        &\kappa(T,\cdot)=0.
      \end{split}
    \right.
  \end{equation}
  where $\delta(t,x)=F(t,x,u_1(t,x))-F(t,x,u_2(t,x))$.  The inequality
  (\ref{equ:beta}) of Lemma \ref{lem:lin-prob} implies that
  \begin{equation}
  \label{equ:star}
    \begin{split}
     \nb{\kappa} =  \nb{\sG u_2-\sG u_1} \leq \tfrac{1}{\beta} \nb{\delta}.
    \end{split}
  \end{equation}
  On the other hand, using inequality \eqref{equ:est-comp-3} of Lemma
  \ref{lem:F-delta}, we obtain
  \begin{equation}%
  \label{equ:double-star}
    \begin{split}
     \nb{\delta} \leq
      \beta e^{-\beta(T-t)}\abs{b(t,x)}\gamma e^{\gamma \s{u_1}\vee
        \s{u_2}}\abs{u_2(t,x)-u_1(t,x) }_0,
    \end{split}
  \end{equation}
  where we remind the reader that
$\s{u}=\sup_{(t,x)\in [0,T]\times \R} u(t,x)$.

If $u$ belongs to the range $\sR \sG$ of $\sG$, then the
 maximum principle
applied to equation \eqref{equ:lin-prob-2} and the non-negativity of
$b$  imply that
  \begin{equation}
    \label{equ:u-above}
    \begin{split}
      \s{u} \leq T \nz{a}+\nz{g}.
    \end{split}
  \end{equation}
Inequalities \eqref{equ:star}, \eqref{equ:double-star} and 
\eqref{equ:u-above} show that
 for $u_1,u_2\in\sR\sG$ 
  \begin{equation}
    \label{equ:sG-Lip}
    \begin{split}
      \nb{\sG u_2-\sG u_1}\leq \frac{1}{\beta} \nz{b} \gamma e^{\gamma
        (T \nz{a}+\nz{g})}\nb{u_2-u_1}.
    \end{split}
  \end{equation}
  It follows that the mapping $\sG$ is a $\nb{\cdot}$-contraction for
  $\beta>0$ large enough. For such $\beta$, the sequence $\seq{u}$,
  where $u_1\in \Cd$ and $u_{n+1}=\sG u_n$, $n\in\N$, converges towards
  some $\hat{u}\in \C$ in $\nb{\cdot}$, and, therefore, also in
  $\nz{\cdot}$.  Our next task is to show that $\hu$ is not only in
  $\C$, but also in $\Cd$ and that it is, indeed, a fixed point of
  $\sG$.

  Let $u\in\sR\sG$ and set $\tu=\sG u$. Then, the inequality
  \eqref{equ:u-above} is satsified by $\tu$ as well. Moreover,
  applying the maximum principle once again to equation
  \eqref{equ:lin-prob-2}, we get that 
 \begin{equation}%
 \label{equ:triple-star}
    \begin{split}
 -\nz{g}-T \nz{b}e^{\gamma(T\nz{a}+\nz{g})} \leq \tu(t,x). 
    \end{split}
\end{equation}
Combining \eqref{equ:triple-star} with inequality \eqref{equ:u-above}
for $\tu$, we obtain
  \begin{equation}
    \label{equ:u-zero}
    \begin{split}
      \nz{\tu}\leq \nz{g}+ T \max\{\nz{a}, \nz{b}e^{\gamma (T\nz{a}+\nz{g})}\}.
    \end{split}
  \end{equation}
  which provides  a
  uniform $\nz{\cdot}$-bound on all elements of the sequence
  $\seq{u}$.  In the remainder of the proof, $D$ will denote a
  generic constant which may depend on $\alpha,T,\gamma, \nzd{h},
  \nzd{a},\nzd{b}$ or $\nzddd{g}$, but is independent of $n$ and $u_1$
  and may change from occurrence to occurrence.  Thanks to the uniform
  bound on $\nz{u_n}$ established in \eqref{equ:u-zero} and recalling
  that $u_{n+1} = \sG u_n \in \Cddd$, the first
  inequality in \eqref{equ:mult-interpol} of Theorem
  \ref{thm:mult-interpol} yields
  \begin{equation}\label{equ:interp-ineq}
    \begin{split}
      \hd{u_{n+1}} &\leq D \hddd{u_{n+1}}^{ \alpha/(2+\alpha)}.
    \end{split}
  \end{equation}
On the other hand, Corollary \ref{cor:hold-trans} implies that
  \begin{equation}
    \label{equ:bd2}
    \begin{split}
      \hddd{u_{n+1}}&\leq D\Big(\hddd{g}+ (
      1+\nzd{h}^{2+\alpha})\nz{u_{n+1}}+
      \hd{a-b e^{\gamma u_n}} \Big),
    \end{split}
  \end{equation}
while by  inequality 
\eqref{equ:est-comp-2} of Lemma \ref{lem:F-delta}, we have
 $\hd{e^{\gamma w}}\leq \gamma e^{\gamma \s{w}}
  \hd{w}$ for any $w\in\Cd$.
So, using once more the uniform bound \eqref{equ:u-zero} on $\nz{u_n}$, we get
\[
\hddd{u_{n+1}}\leq D(1+\nzd{b} \nzd{e^{\gamma u_n}})\leq D(1+\hd{u_n}).
\]
 Therefore, by \eqref{equ:interp-ineq}, we have
  \begin{equation}
    \nonumber 
    \begin{split}
      \hd{u_{n+1}}\leq D(1+\hd{u_{n}})^{\alpha/(2+\alpha)},
      \text{ for $n\in\N$}.
    \end{split}
  \end{equation}
  This fact implies that the sequence $x_n= \hd{u_n}^{(2+\alpha)/\alpha}$
  satisfies \eqref{equ:recurr} with $A= B=D^{(2+\alpha)/\alpha}$ so
  that   
by Lemma \ref{lem:sequence}, 
\[ \sup_n \hd{u_n}<\infty.\] It is not difficult to see that the
$\nz{\cdot}$-closure of a $\nzd{\cdot}$-bounded subset of $\C$ is, in
fact, a subset of $\Cd$.  Consequently, the $\nz{\cdot}$-limit
$\hat{u}$ of $\seq{u}$ is in $\Cd$, and in particular, $\sG \hat u$ is
well-defined. Inequality \eqref{equ:sG-Lip} yields
\[ \nz{\sG \hat u - u_{n+1}}= \nz{\sG \hat u - \sG u_n}\leq D \nz{\hat
  u - u_n},\] so that $\sG \hat u=\lim_n u_n=\hat u$, i.e.,
$\hat{u}$ solves \eqref{equ:semi}.  Moreover,  any solution
$u$ must satisfy \eqref{equ:u-above} and the relation \eqref{equ:sG-Lip}
 guarantees that $\hat{u}$ is unique in the
class $\Cd$.

To establish the bounds in \eqref{equ:bounds-nz-u} we pick two
continuous functions $\ha,\hb:[0,T]\to \R$ and consider the function
\[w(t,x)=A(t)-\frac{1}{\gamma} \log \Big[ \gamma B(t)+\exp(-\gamma G)
\Big] ,\] where $A(t)=\int_t^T \ha(u)\, du$, $B(t)=\int_t^T \hb(u)
e^{\gamma A(u)}\, du\geq 0$, and $G$ is an arbitrary constant. 
As the reader can easily check, we have
  \begin{equation}
    \nonumber
    \left\{   \begin{split}
       0= &w_t+\tot w_{xx}+ h(t,x) w_x+a(t,x)- b(t,x) e^{\gamma w}\\ &-(a(t,x)-\ha(t))
        -(\hb(t)-b(t,x)) e^{\gamma w(t,x)}\\
        w(T&,\cdot)=G.
      \end{split}
    \right.
  \end{equation}
  For different choices of  functions $\ha,\hb$ and the constant $G$,
  $w$ will be either a sub- or a supersolution of
  (\ref{equ:semi}). Indeed, for $G=\nz{g}$, $\ha(t)=\nz{a}$ and
  $\hb(t)=0$, $w$ becomes a supersolution, yielding the upper bound in
  \eqref{equ:bounds-nz-u}. Similarly, for $G=-\nz{g}$, $\ha(t)=0$ and
  $\hb(t)=\nz{b}$, $w$ is a subsolution, and so, the lower bound in \eqref{equ:bounds-nz-u} holds, too. 

  The last item on our list is the  $\hddd{\cdot}$-bound
  \eqref{equ:bound-ddd}. Estimate
\eqref{equ:hold-trans}, along with Lemma \ref{lem:F-delta} and the just
established \eqref{equ:bounds-nz-u}, yields 
\begin{equation}
\nonumber
  \begin{split}
    \hddd{u} &\peq \hddd{g}+ \hd{a}+\hd{b e^{\gamma u}}+
    (1+\nzd{h}^{2+\alpha})\nz{u} \\
    &\peq \hddd{g}+ \hd{a}+\hd{b} e^{\gamma \s{u}}+\nz{b}
    \hd{e^{\gamma u}}+
    (1+\nzd{h}^{2+\alpha})\nz{u} \\
    &\peq \hddd{g}+ \hd{a}+\hd{b} e^{\gamma(T\nz{a}+\nz{g})}
    +\nz{b}
    \gamma e^{\gamma(T\nz{a}+\nz{g})} \hd{u}+
    (1+\nzd{h}^{2+\alpha})\nz{u}. \\
   \end{split}
\end{equation}
The inequality \eqref{equ:bound-ddd} now follows directly from 
Corollary \ref{cor:onen}.
\end{proof}
Finally, we integrate Lemma \ref{lem:lin-prob} and Proposition
\ref{pro:semilinear} and discuss the difficulties in the 
standard verification argument. \begin{proposition}
  \label{pro:PDDE}
  For $\ld\in\Cds$ and $g\in\Cdddrs$ there exists a unique function
  $\ul\in\Cddds$ which solves the Cauchy problem (\ref{equ:PDDE}).
  Moreover, the stochastic process $\prf{\pil_t}$ with Markov
  representative $\pil=\tfrac{1}{\gamma}\ld -\ul_x$ is in $\AAl$ and
  attains the maximum in the utility maximization problem
  \eqref{equ:vi-ld}.
\end{proposition}
\begin{proof}
  We note, first, that the $n=1$-slice of the equation
  \eqref{equ:PDDE} has the form \eqref{equ:lin-prob} with
  $h(t,x)=-\ld(t,x)\in\Cd$, $a(t,x)=\tfrac{1}{2\gamma}
  \ld^2(t,x)\in\Cd$.  By
  Lemma \ref{lem:lin-prob}, it admits a unique $\Cddd$-solution
  $u(\cdot,\cdot,1)$. Therefore, the $n=0$-slice is of the form
  \eqref{equ:semi} with $h(t,x)=-\ld(t,x)$, $a(t,x)= \tfrac{1}{\gamma}
  (\tot \ld^2(t,x)+\mu)$ and $b(t,x)=\tfrac{\mu}{\gamma}e^{-\gamma
    u(t,x,1)}$; the assumptions placed on $\ld$ and the $\Cd$-property
  of $u(\cdot,\cdot,1)$ imply that Proposition \ref{pro:semilinear} can be
  applied. Hence, \eqref{equ:PDDE} admits a unique solution
  $\ul\in\Cddds$.

  Having established existence and regularity of the solution $u$ of
  \eqref{equ:PDDE}, we turn to the second statement - namely, that the
  process $\pil$, defined in the statement,  is optimal for
  \eqref{equ:vi-ld}. First, we consider the function $v:[0,T]\times
  \R\times\R\times\set{0,1}\to (-\infty,0)$, given by
  \[ v(t,\xi,x,n)=-e^{-\gamma\big(\xi+\ul(t,x,n)\big)}.\] The reader will
  easily check that $v$ is a classical solution of the following PDDE
  - the formal HJB equation for the utility-maximization problem
  \eqref{equ:vi-ld}:
  \begin{equation}
    \label{equ:HBJ}
    \left\{
      \begin{split}
        & 0=v_t + \sup_{\pi\in\R} \Big( \tot v_{xx} + \pi \ld
        v_{\xi} + \tot \pi^2 v_{\xi\xi}+\pi v_{x\xi} +\mu v_n
        \Big)\\
        & v(T,\cdot,\cdot) =-e^{-\gamma\big(\xi+g(x,n)\big)},
      \end{split}
    \right.
  \end{equation}
  with all regularity inherited from $\ul$. The equation \eqref{equ:HBJ}
  is the Hamilton-Jacobi-Bellman equation for the control problem
  \eqref{equ:vi-ld}, where the variables $\xi,x,n$ correspond to the
  wealth process $\int_0^{\cdot} \pil_u\, d\Sl_u$, the Brownian motion
  $B$ and the jump process $N$, respectively.  The standard
  verification procedure (see, e.g., the ideas in the proof of Theorem
  8.1, p.~141 in \cite{FleSon93}) can be used to show that $v$ is
  indeed the value function of the utility-maximization problem and
  that the form of the optimal portfolio can be recognized as the
  optimal value of the parameter $\pi$ in the maximization in
  \eqref{equ:HBJ}. The so-obtained $\pi$ is admissible since it
  is uniformly bounded. The usual difficulty one encounters in the
  course of the verification procedure - namely, the one involved in
  establishing the martingale property of certain local martingales -
  is easily dealt with here. One simply needs to observe that $v$ and
  $\ul_x$ are uniformly bounded and use the square-integrability of
  $\pi\in\AAl$.
\end{proof}
\subsection{Stability of the optimal portfolio}
Having established the existence and uniqueness of the regular
solution to \eqref{equ:PDDE}, we turn to the study of the map
$\ld\mapsto \ul$ in order to prove the statements in the second part
of Theorem \ref{thm:stability}. We start with some preliminary growth
estimates.
\begin{proposition}
\label{pro:growth}
Let $u\in\Cddds$ be the unique solution to \eqref{equ:PDDE}. 
Define 
\begin{equation}
  \label{equ:def-M}
   \begin{split}
 M_0=\gamma \nz{g}+\tot T \nz{\ld}^2+\mu T,\quad 
 \Ma=\gamma \nzddd{g}+(T+1)(\tot \nzd{\ld}^2+\mu).
   \end{split}
\end{equation}
Then, 
\begin{align}
\label{equ:bd-all}
\nz{\ul} & \peq M_0,\text{ and } \hddd{\ul} \peq  
\Ma^{2+\alpha/2} e^{(2+\alpha) M_0}.
\end{align}
\end{proposition}
\begin{proof}
  Set $L_0=\nz{\ld}$, $\La=\nzd{\ld}$ and  $G=\nzddd{g}$, and let $M_0,\Ma$ be as in
  \eqref{equ:def-M}.  We start with the $n=1$-slice first; to make the
  notation more palatable, we write simply $\bu$ for
  $\ul(\cdot,\cdot,1)$ and $\bld$ for $\ld(\cdot,\cdot,1)$ so that $\bu$
  satisfies the linear equation (\ref{equ:lin-prob}) with
  $a=\tfrac{1}{2\gamma} \bld^2$ and $h=-\bld$. With
  $\overline{g}=g(\cdot,1)$, the maximum principle implies that
\begin{equation}
   \label{equ:supn}
   \begin{split}
     \nz{\bu}\leq T\nz{{a}}+\nz{\bar{g}}\leq  \tfrac{1}{2\gamma}
     TL_0^2+ \nz{g}\peq M_0.
   \end{split}
\end{equation}
  Corollary \ref{cor:hold-trans} yields
\begin{equation}
\label{equ:est-bu}
   \begin{split}
     \hddd{\bu}&\peq G +\tfrac{1}{2\gamma} \hd{\bld^2}
     +\big(1+\nzd{\bld}^{2+\alpha}\big)\nz{\bu} \\
     &\peq
     G+\La^2+(1+\La^{2+\alpha})(\tot T \La^2+\gamma G)\\
&\peq (1+\La^{2+\alpha})(G+(1+T)\La^2)\peq \Ma(1+\La^{2+\alpha})\peq \Ma^{2+\alpha/2},
   \end{split}
\end{equation}
where the last inequality is a consequence of the fact that the
function $\rho\mapsto (1+x^{\rho})^{1/\rho}$ is nonincreasing in
$\rho$ for $x\geq 0$ and $\rho>0$.  Also, due to the first inequality
in \eqref{equ:mult-interpol} and the just established \eqref{equ:supn}
and \eqref{equ:est-bu}, we have 
\begin{equation}\label{equ:ineq1}
    \begin{split}
      \hd{\bu} \preceq \hddd{\bu}^{\tfrac{\alpha}{2+\alpha}}
      \abs{\bu}_0^{\tfrac{2}{2+\alpha}} \preceq (M_\alpha^{2 +
        \tfrac{\alpha}{2}})^{\tfrac{\alpha}{2+\alpha}} \,
      M_0^{\tfrac{2}{2+\alpha}} \preceq M_\alpha^{2 +
        \tfrac{\alpha}{2}}. 
    \end{split}
\end{equation}

We move on to the
$n=0$-slice, and write $\uu(\cdot,\cdot)$ for $\ul(\cdot,\cdot,0)$
and $\uld$ for $\ld(\cdot,\cdot,0)$ so that $\uu$ satisfies the
semi-linear Cauchy problem \eqref{equ:semi} with $a=\tfrac{1}{2\gamma}
\uld^2+\tfrac{\mu}{\gamma}$, $h=-\uld$ and $b=\tfrac{\mu}{\gamma}
\exp(-\gamma \bu)$. Note that $\bu$ also solves equation
\eqref{equ:semi} with $a = \tfrac{1}{2\gamma}\bld^2, h = - \bld$ and
$b=0$, so that the estimate \eqref{equ:bounds-nz-u} implies
$\nz{e^{-\gamma\bu}}=e^{-\gamma \i{\bu}}\leq e^{\gamma \nz{\bg}}$, as
well as 
\begin{equation}
\nonumber 
   \begin{split}
     \nz{\uu} & \leq \max\Big(\tfrac{1}{2\gamma} T (L_0^2+2\mu)+
     \nz{g},\tfrac{1}{\gamma} 
     \log(e^{\gamma \nz{g}}+ \mu T\nz{e^{-\gamma \bu}} )\Big)\\
     & \leq \nz{g}+\tfrac{1}{\gamma} 
       \max\Big( T (\tot L_0^2+\mu),\log(1+ \mu T)\Big).\\
   \end{split}
\end{equation}
This result, together with \eqref{equ:supn} and the fact that $\mu T\geq
\log(1+\mu T)$, yields the $\nz{\cdot}$-bound in \eqref{equ:bd-all}.
By inequality \eqref{equ:est-comp-2} of
Lemma \ref{lem:F-delta} and estimates \eqref{equ:bounds-nz-u} and
\eqref{equ:ineq1}, we have
\begin{equation}\label{equ:ineq2}
   \begin{split}
     \hd{e^{-\gamma \bu}}&\leq \gamma e^{\gamma G} \Ma^{2+\tfrac{\alpha}{2}}.
\end{split}
\end{equation}
Hence, using again the estimate \eqref{equ:bounds-nz-u} for $\bu$, the
just established first inequality in \eqref{equ:bd-all} and
\eqref{equ:ineq2}, we obtain
\begin{equation}
    \nonumber 
    \begin{split}
      \hd{e^{-\gamma \bu} e^{\gamma \uu}} 
      & \preceq e^{\gamma G} e^{M_0} \Ma^{2+\tfrac{\alpha}{2}} +
      e^{\gamma G} e^{M_0} \hd{\uu} = e^{\gamma G + M_0}
      (\Ma^{2+\tfrac{\alpha}{2}} + \hd{\uu}).
    \end{split}
\end{equation} 
So, due to Corollary \ref{cor:hold-trans}, 
\begin{equation}
    \nonumber 
    \begin{split}
      \hddd{\uu} & \preceq \hddd{\ug} + \frac{1}{2\gamma} \hd{\uld^2}
      + \tfrac{\mu}{\gamma} \hd{e^{-\gamma \bu} e^{\gamma \uu}} +
      (1+\nzd{\uld}^{2+\alpha})\nz{\uu} \\
      & \preceq \tfrac{1}{\gamma} (\gamma G + \tfrac{1}{2}L_\alpha^2 +
      \mu) + \hd{e^{-\gamma \bu} e^{\gamma \uu}} +
      (1+L_\alpha^{2+\alpha})\nz{\uu}\\
      & \preceq M_0 + e^{M_0 + \gamma G} \Ma^{2+\tfrac{\alpha}{2}} +
      e^{M_0 + \gamma G} \hd{\uu} + (1+L_\alpha^{2+\alpha})\nz{\uu}.\\
    \end{split}
\end{equation}
Therefore, thanks to  Corollary \ref{cor:onen} and the first
inequality in \eqref{equ:bd-all},
\begin{equation}
    \nonumber 
    \begin{split}
      \hddd{\bu} & \preceq e^{\gamma G + M_0} + \left(e^{(\gamma G +
          M_0)(1+\tfrac{\alpha}{2})} + 1 + L_\alpha^{2+\alpha}\right)\nz{\uu}.
    \end{split}
\end{equation}
Thus, using once more the decrease of the function $\rho \mapsto
(1+x^\rho)^{1/\rho}, x, \rho >0$ and the first inequality of
\eqref{equ:bd-all}, we get 
\begin{equation}
    \nonumber 
    \begin{split}
      \hddd{\uu} \preceq e^{(2+\alpha)M_0} \Ma^{2+\tfrac{\alpha}{2}} +
      M_0 \Ma^{2+\tfrac{\alpha}{2}} \preceq e^{(2+\alpha)M_0} \Ma^{2+\tfrac{\alpha}{2}}.
    \end{split}
\end{equation} 
The last inequality, along with \eqref{equ:est-bu}, completes the proof. 
\end{proof}
The multiplicative interpolation inequalities of
Theorem \ref{thm:mult-interpol} yield the following corollary:
\begin{corollary}
\label{cor:nzdd-bd}
Let $\ul\in\Cddds$ be the unique solution to \eqref{equ:PDDE}, and let
$M_0$ and $\Ma$ be as in \eqref{equ:def-M}. Then, the following
estimates hold
\begin{align}
   \label{equ:nzdd-bds}
   \hd{\ul} &\peq \Ma^{2+\alpha/2 }e^{\alpha M_0},& \nz{\ul_x}  &\peq \Ma^{2+\alpha/2} e^{M_0}\\
\notag   \hd{\ul_x} & \peq \Ma^{2+\alpha/2} e^{(1+\alpha) M_0},&
   \nzd{\ul_x} & \peq \Ma^{2+\alpha/2} e^{(1+\alpha) M_0}.
   \end{align}
\end{corollary}
We continue with a stability estimate for the nonlinear, $n=1$-slice. 
\begin{proposition}
\label{pro:diff-1}
For $g\in\Cdddr$, $a^{(k)}\in\Cd$, $h^{(k)}\in\Cd$ and $b^{(k)}\in\Cd$, 
let $u^{(k)}$, $k=1,2$ be the unique $\Cddd$-solution to
  \begin{equation}
    \label{equ:semi-i}
    \left\{   
\begin{split}
  &0=\ukey_t+\tot \ukey_{xx}+ \hkey \ukey_x+\akey-
  \bkey e^{\gamma \ukey} \text{ on } [0,T)\times \R, \\
  &\ukey(T,\cdot)=g,
      \end{split}
    \right.
  \end{equation}
and let 
\[
D=1+\max\Big(\nzd{\bone}^{1+\alpha/2}, \nzd{\uone_x}, \nzd{\utwo},
\nzd{\htwo}^{2+\alpha}\Big),\ P=\gamma e^{\gamma \s{\uone}\vee
  \s{\utwo}}. \] Then, 
\begin{align}%
\label{equ:nz}
\nz{\utwo-\uone}&\leq T e^{T \nz{\bone} P }\Big(\nz{\adel}+
D\nz{\hdel}+P\nz{\bdel} \Big),\text{ and }\\
\label{equ:delt}
\hddd{\utwo-\uone}&\peq \hd{\adel}+D
\nzd{\hdel} + PD \nzd{\bdel}+
P^{1+\alpha/2} D\nz{\utwo-\uone},
\end{align}
where 
$\adel = a^{(2)}-a^{(1)}, \hdel = h^{(2)}-h^{(1)}$ and $\bdel =
  b^{(2)}-b^{(1)}$.
\end{proposition}
\begin{proof}
  Set $\udel = u^{(2)}-u^{(1)}$  and subtract the equations \eqref{equ:semi-i} with $k=2$
  and $k=1$, respectively, to obtain that $\udel$ is a
  $\Cddd$-solution to the following semilinear Cauchy problem
  \begin{equation}
    \label{equ:semi-delta}
    \left\{   
\begin{split}
  0=\udel_t&+\tot \udel_{xx}+ h^{(2)} \udel_x+\hdel
  u^{(1)}_x + \adel\\ & - \bdel e^{\gamma u^{(2)}}-b^{(1)} (e^{\gamma
    u^{(2)}}-e^{\gamma u^{(1)}}),
  \text{ on } [0,T)\times \R, \\
  \udel(T,\cdot)&=0.
      \end{split}
    \right.
  \end{equation}
To alleviate the notation, we set
\[ 
K_0=\nz{\hdel\uone_x+\adel-\bdel
  e^{\gamma\utwo}} \text{ and }
K_{\alpha}=\hd{\hdel\uone_x+\adel-\bdel
  e^{\gamma\utwo}}.\] 
First, we establish the $\nz{\cdot}$-bound in \eqref{equ:nz}. Let
$U(t)=\sup_{x\in\R} \abs{\udel(t,x)}$ 
so that, by the mean value theorem, we have 
\[  \abs{e^{\gamma \utwo(t,x)}-e^{\gamma \uone(t,x)}}
\leq 
P\,  U(t),\text{ for all }
(t,x)\in [0,T]\times \R.\]
Hence, the maximum principle with equation \eqref{equ:semi-delta} implies 
\begin{equation}\label{equ:mp-ineq}
    \begin{split}
      \abs{\udel(t,x)} \leq \int_t^T \big( K_0+\nz{\bone}P\, U(s)\big) \, ds. 
    \end{split}
\end{equation}
We can, therefore, employ 
Gronwall's inequality (see, e.g., \cite{RevYor99}, p.~543) to obtain
\[ U(t)\leq (T-t) K_0 e^{(T-t) \nz{\bone} P }\leq T K_0 e^{T
  \nz{\bone} P}.\] Noting that $ K_{0} \leq D \nz{\hdel}+ \nz{\adel}+
P \nz{\bdel}$, we conclude that \eqref{equ:nz} holds. 

Turning to the $\hddd{\cdot}$-bound \eqref{equ:delt}, we use Corollary
\ref{cor:hold-trans} and Lemma \ref{lem:F-delta} to conclude that
\begin{equation}%
\label{equ:ioio}
    \begin{split}
      \hddd{\udel} &\peq \hd{\hdel \uone_x + \adel - \bdel e^{\gamma
          \utwo}- \bone (e^{\gamma \utwo}-e^{\gamma \uone})} + \Big(
      1+\nzd{\htwo}^{2+\alpha}
      \Big) \nz{\udel}\\
      &\peq K_{\alpha} +\nz{\bone} \hd{e^{\gamma \utwo}-e^{\gamma
          \uone}} +\hd{\bone} \nz{e^{\gamma \utwo}-e^{\gamma \uone}}+
      D \nz{\udel}\\
      &\peq K_{\alpha} +\nz{\bone} P\hd{\udel}+ P \nzd{\utwo}
      \nz{\udel} +\hd{\bone} P \nz{\udel}+ D \nz{\udel}.
    \end{split}
\end{equation} 
Therefore, by Corollary \ref{cor:onen}, we have
\begin{equation}
   \label{equ:asdf}
   \begin{split}
 \hddd{\udel}&\peq K_{\alpha} + \Big(P\nz{\bone})^{1+\alpha/2}+
P\nzd{\utwo}+P\hd{\bone}+ D\Big) \nz{\udel}\\ & \peq
K_{\alpha} + P^{1+\alpha/2}D \nz{\udel}.
   \end{split}
\end{equation}
Finally, due to Lemma \ref{lem:F-delta}, 
$ K_{\alpha} \peq D \nzd{\hdel}+ \hd{\adel}+ PD \nzd{\bdel}$,
which implies \eqref{equ:delt}. 
\end{proof}
In conjunction with the multiplicative interpolation inequalities of
Theorem \ref{thm:mult-interpol}, our final result, Proposition
\ref{pro:dif}, yields the remaining statements of Theorem
\ref{thm:stability}.
 \begin{proposition}
\label{pro:dif}
  For $g\in\Cddds$ and $k=1,2$, let $\ukey\in\Cddds$ be the unique
  solution to \eqref{equ:PDDE}, corresponding to $\ld=\ldkey\in\Cds$.
  Then, 
\begin{equation}
    \nonumber 
    \begin{split}
      \nz{\utwo-\uone}&\peq T\, \Ma^{2+2\alpha} e^{e^{2M_0+1}} \nz{\ldtwo-\ldone},\text{ and }\\
      \hddd{\utwo-\uone}&\peq \Ma^{7+3\alpha} e^{e^{2M_0+2}}
      \nzd{\ldtwo-\ldone},
    \end{split}
\end{equation}
where
\begin{equation}
\nonumber 
   \begin{split}
M_0&= \gamma \nz{g}+\mu T + \tot T
\max(\nz{\ldone},\nz{\ldtwo})^2,\text{ and }\\
 \Ma&=
\nzddd{g}+(1+T)\Big(\mu+\tfrac{1}{2}\max(\nzd{\ldone},\nzd{\ldtwo})^2\Big).
   \end{split}
\end{equation}
\end{proposition}
\begin{proof}
  Let $L_0=\max(\nz{\ldone},\nz{\ldtwo})$,
  $\La=\max((\nzd{\ldone},\nzd{\ldtwo})$, $G_0=\nz{g}$ and
  $G=\nzddd{g}$. Then, 
\[M_0= \gamma G_0+\mu T + \tot T L_0^2,\ 
\Ma= G+(1+T)(1+\La^2).\]
 Just like in the
  proof of Proposition \ref{pro:growth}, we deal with the $n=1$-slice
  first; again, to make the notation more palatable, we write $\bukey$
  for $\ukey(\cdot,\cdot,1)$ and $\bldkey$ for $\ldkey(\cdot,\cdot,1)$
  so that $\bukey$ satisfies the linear equation (\ref{equ:semi-i}) with
  $\akey=\tfrac{1}{2\gamma} (\bldkey)^2$, $\hkey=-\bldkey$ and $\bkey=0$.
  Also, set $\budel=\butwo-\buone$, $\Delta_0=\nz{\ldtwo-\ldone}$ and
  $\Dela=\nzd{\ldtwo-\ldone}$.  Then, $\budel$ solves the Cauchy
  problem \eqref{equ:semi-delta} with $\hdel = -\bldtwo+ \bldone$,
  $\adel = \tfrac{1}{2\gamma}((\bldtwo)^2 - (\bldone)^2)$ and $\bdel =
  0$. Inequality \eqref{equ:mp-ineq}
  and Corollary \ref{cor:nzdd-bd} yield 
\begin{equation}
\begin{split}
\label{equ:nz-2}
\nz{\budel}&\preceq 
 T(L_0 \Delta_0+\nz{\buone_x} \Delta_0)
 \peq T \Delta_0(L_0 + \Ma^{2+\alpha/2} e^{M_0}) \peq T \Ma^{2+\alpha/2} e^{M_0} \Delta_0.
   \end{split}
\end{equation}
So, by Corolary \ref{cor:hold-trans}, estimate
 \eqref{equ:nzdd-bds} and the decrease of the function $\rho\mapsto (1+x^\rho)^{1/\rho}$,
we have 
\begin{equation}
\label{equ:delt-2}
   \begin{split}
     \hddd{\budel}&\peq \nzd{\buone_x}\Dela + \Dela L_\alpha +
     (1+\La^{2+\alpha}) \nz{\budel}\\
     & \peq \Big( \Ma^{2 + \alpha/2}e^{(1+\alpha)M_0} + \Ma +
     \Ma^{1+\alpha/2} T  \Ma^{2+\alpha/2} e^{M_0} \Big)\Dela\\
     & \peq \Ma^{4+\alpha} e^{(1+\alpha)M_0} \Dela.
   \end{split}
\end{equation}
Consequently, the multiplicative interpolation inequalities of Theorem
\ref{thm:mult-interpol} imply 
\begin{equation}
\label{equ:delt-3}
   \begin{split}
     \hd{\budel}&\peq 
T^{2/(2+\alpha)} \Ma^{2+2\alpha} e^{(1+\alpha) M_0} \Dela.
   \end{split}
\end{equation}

We proceed with the $n=0$-slice, and
write $\uukey(\cdot,\cdot)$ for $\ukey(\cdot,\cdot,0)$ and $\uldkey$
for $\ldkey(\cdot,\cdot,0)$ so that $\uukey$ satisfies the semi-linear
Cauchy problem \eqref{equ:semi} with $\akey=\tfrac{1}{2\gamma}
(\uldkey)^2+\tfrac{\mu}{\gamma}$, $\hkey=-\uldkey$ and
$\bkey=\tfrac{\mu}{\gamma} \exp(-\gamma \bukey)$.  Set
\[ D= \Ma^{2+2\alpha} e^{(1+2\alpha)M_0}\text{ and } P= \gamma
e^{\gamma \s{g}+ \tot T L_0^2},\] so that, by Proposition
\ref{pro:growth}, Corollary \ref{cor:nzdd-bd} and inequality
\eqref{equ:u-above}, $D$ and $P$ dominate (in the $\peq-$sense) the $D$ and $P$ of
Proposition \ref{pro:diff-1} with the current choice of $\bone, \uone,
\utwo$ and $\htwo$. 

Set $\uudel=\uutwo-\uuone$. Then $\uudel$ satisfies the equation
\eqref{equ:semi-delta} with $\uone = \uuone, \htwo = - \uldtwo, \hdel = -\uldtwo +
\uldone, \adel = \tfrac{1}{2\gamma}((\uldtwo)^2 - (\uldone)^2), \bdel
= \tfrac{\mu}{\gamma} (\exp(-\gamma \butwo) - \exp(-\gamma \buone))$
and $\bone = \tfrac{\mu}{\gamma} \exp(-\gamma \buone)$.
Inequality \eqref{equ:nz} of 
Proposition \ref{pro:diff-1} and the lower bound on $\buone$ in
\eqref{equ:bounds-nz-u} give us 
\begin{equation}
    \nonumber 
    \begin{split}
      \nz{\uudel} & \peq T e^{T\nz{\bone} P} (\tfrac{1}{2\gamma}
      \nz{(\uldtwo)^2 - (\uldone)^2 } + D \nz{-\uldtwo +
\uldone} + P \nz{e^{-\gamma \buone} - e^{-\gamma \butwo}})\\
           & \peq  T e^{T\mu e^{2\gamma G_0 + \tfrac{1}{2}
               TL_0^2}}(\Delta_0 L_0 + D \Delta_0 + P \nz{e^{-\gamma \buone} - e^{-\gamma \butwo}}).
    \end{split}
\end{equation}
On the other hand, thanks to Lemma \ref{lem:F-delta},  the lower bound
on $\buone$ and $\butwo$ in
\eqref{equ:bounds-nz-u} and inequality \eqref{equ:nz-2}, we have  
\begin{equation}
    \nonumber 
    \begin{split}
      \nz{e^{-\gamma \butwo} - e^{-\gamma \buone}} 
      & \peq \gamma e^{\gamma G_0} \nz{\budel} \peq \gamma e^{\gamma G_0} T \Ma^{2+\alpha/2} e^{M_0}
      \Delta_0. 
    \end{split}
\end{equation}
Combining the last two inequalities we, obtain
\begin{equation}
\begin{split}
\label{equ:nz-3}
\nz{\uudel}& \peq 
          T e^{T\mu e^{2\gamma G_0 + \tfrac{1}{2}
               TL_0^2}} (L_0 + D + P e^{\gamma G_0} T
           \Ma^{2+\alpha/2} e^{M_0})\Delta_0 \\
           & \peq 
           T e^{T\mu e^{2\gamma G_0 + \tfrac{1}{2}
               TL_0^2}} (\Ma^{2+2\alpha}e^{(1+2\alpha)M_0} + P e^{\gamma G_0} T \Ma^{2+\alpha/2} e^{M_0})\Delta_0 \\
           & \peq 
           T e^{T\mu e^{2\gamma G_0 + \tfrac{1}{2}
               TL_0^2}} \Ma^{2+2\alpha}e^{(1+2\alpha)M_0} (1+
           e^{2\gamma G_0 + \tfrac{1}{2} TL_0^2} T)\Delta_0 \\
           & \peq 
           T (1+T) e^{T\mu e^{2\gamma G_0 + \tfrac{1}{2}
               TL_0^2}} \Ma^{2+2\alpha}e^{(1+2\alpha)M_0} 
           e^{2\gamma G_0 + \tfrac{1}{2} TL_0^2}\Delta_0 \\
  &\peq 
T\Ma^{2+2\alpha} e^A\Delta_0,
   \end{split}
\end{equation}
where $A=\log(1+\mu T)+M_0+\mu T e^{2\gamma G_0+1/2 T L_0^2 }
+(1+2\alpha )M_0 + 2\gamma G_0+\tot T L_0^2$. Since $\tfrac{x}{e^x}\leq
\tfrac{1}{e}$, for all $x\geq 0$ and $\log(1+\mu T) \le \mu T$, we have
\begin{equation}\label{equ:bd-A} 
   \begin{split}
A & = \log(1+\mu T)+M_0+\mu T e^{\gamma G_0 + M_0 - \mu T }
+(1+2\alpha )M_0 + M_0 + \gamma G_0 - \mu T\\
&\peq M_0 + e^{\gamma G_0+M_0}\tfrac{\mu T}{e^{\mu T}}
+(2+2\alpha )M_0 + \gamma G_0 
\peq e^{2 M_0+1}.
\end{split}
\end{equation}
It remains to estimate the $\hddd{\cdot}$-seminorm of $\uudel$. 
By \eqref{equ:supn},\eqref{equ:est-bu} and Theorem \ref{thm:mult-interpol}, 
we have
\begin{equation}\label{equ:ineq-hd-butwo}
   \begin{split}
1+     \hd{\butwo}\peq 1+(\Ma^{2+\alpha/2})^{\alpha/(2+\alpha)}
     M_0^{2/(2+\alpha)}\peq 1+\Ma^{1+\alpha}\peq \Ma^{1+\alpha}.
   \end{split}
\end{equation}
Applying Corollaries \ref{cor:hold-trans} and \ref{cor:nzdd-bd},
and using the decrease of the function $\rho\mapsto (1+x^\rho)^{1/\rho},
x, \rho >0$, we get 
\begin{equation}\label{equ:ineq-hddd-udel-1}
    \begin{split}
      \hddd{\uudel} & \peq 
      \Dela \nzd{\uuone_x} + \La \Dela + e^{\gamma G_0}(1+\hd{\uutwo})
      \nzd{ e^{-\gamma \butwo} - e^{-\gamma \buone}} \\
      & \qquad + \nzd{e^{-\gamma
        \buone}} \nzd{ e^{\gamma \uutwo} - e^{\gamma \uuone}} +
    (1+\La^{2+\alpha}) \nz{\uudel}\\
    & \peq 
      \Dela \Ma^{2+\alpha/2} e^{(1+\alpha)M_0} + \Ma \Dela + e^{\gamma
        G_0} (1+\Ma^{2+\alpha/2} e^{\alpha M_0})
      \nzd{ e^{-\gamma \butwo} - e^{-\gamma \buone}} \\
      & \qquad + \nzd{e^{-\gamma
        \buone}} \nzd{ e^{\gamma \uutwo} - e^{\gamma \uuone}} +
    \Ma^{1+\alpha/2} \nz{\uudel}.\\
    \end{split}
\end{equation}
On the other hand, by Lemma \ref{lem:F-delta}, the lower bound on $\buone$ and $\butwo$ in
\eqref{equ:bounds-nz-u}, and inequalities \eqref{equ:nz-2},
\eqref{equ:delt-2} and
\eqref{equ:ineq-hd-butwo}, we have
\begin{equation}\label{equ:ineq-budel}
    \begin{split}
      \nzd{ e^{-\gamma \butwo} - e^{-\gamma \buone}}
      & \peq e^{\gamma G_0} \big(\nz{\budel} (1+\hd{\butwo}) + \hd{\budel}\big)\\
      & \peq e^{\gamma G_0} \big(T \Ma^{2 + 3\alpha/2} e^{M_0} +
      T^{2/(2+\alpha)} \Ma^{2+2\alpha} e^{(1+\alpha)M_0}  \big)\Dela\\
      & \peq \Ma^{4+3\alpha/2} e^{(2+\alpha)M_0} \Dela.
    \end{split}
\end{equation}
Moreover, due to Lemma \ref{lem:F-delta} and Corollary
\ref{cor:nzdd-bd}, $\nzd{e^{-\gamma\buone}} \peq \Ma^{2+\alpha/2}
e^{(1+\alpha)M_0}$ and, with the upper bound in
\eqref{equ:bounds-nz-u} and inequalities \eqref{equ:nz-3} and
\eqref{equ:bd-A} in mind,  
\begin{equation}
    \nonumber 
    \begin{split}
      \nzd{e^{-\gamma \uutwo} - e^{-\gamma \uuone}} 
      & \peq e^{M_0} \big(T \Ma^{2+2\alpha} e^{e^{2M_0+1}} \Delta_0
      (1+\Ma^{2+\alpha/2} e^{\alpha M_0}) + \hd{\budel}
      \big)\\
         & \peq \Ma^{5+5\alpha/2}  e^{e^{2M_0+1}} e^{(1+\alpha)
           M_0} \Delta_0 + e^{M_0}\hd{\budel}.\\
    \end{split}
\end{equation} 
So, 
\begin{equation}\label{equ:ineq-buone-budel}
    \begin{split}
      \nzd{e^{-\gamma\buone}}   \nzd{e^{-\gamma \uutwo} - e^{-\gamma
          \uuone}} 
      \peq \Ma^{7+3\alpha} e^{e^{2M_0+1}} e^{(2+2\alpha)M_0} \Delta_0
      + \Ma^{2+\alpha/2} e^{(2+\alpha) M_0} \hd{\budel}.
    \end{split}
\end{equation}
Combining inequalities \eqref{equ:ineq-hddd-udel-1},
\eqref{equ:ineq-budel} and \eqref{equ:ineq-buone-budel}, we get  
\begin{equation}\label{equ:ineq-hddd-udel-2}
    \nonumber 
    \begin{split}
      \hddd{\uudel}    
      & \peq 
      \Ma^{2+\alpha/2} e^{(1+\alpha)M_0} \Dela + \Ma^{2+\alpha/2} e^{(1+\alpha )M_0}
      \Ma^{4+3\alpha/2} e^{(2+\alpha)M_0} \Dela\\
      & \qquad  + \Ma^{7+3\alpha} e^{e^{2M_0+1}} e^{(2+2\alpha)M_0} \Dela
      + \Ma^{2+\alpha/2} e^{(2+\alpha) M_0} \hd{\budel} + 
    \Ma^{1+\alpha/2} \nz{\uudel}\\
         & \peq \Ma^{7+3\alpha} e^{e^{2M_0+1}} e^{(3+2\alpha)M_0} \Dela
      + \Ma^{2+\alpha/2} e^{(2+\alpha) M_0} \hd{\budel} + 
    \Ma^{1+\alpha/2} \nz{\uudel}.\\
    \end{split}
\end{equation}
Thanks to Corollary \ref{cor:onen}, we obtain
\begin{equation}\label{equ:ineq-hddd-udel-3}
    \nonumber 
    \begin{split}
      \hddd{\uudel}    
         & \peq \Ma^{7+3\alpha} e^{e^{2M_0+1}} e^{(3+2\alpha)M_0} \Dela
      + \big( (\Ma^{2+\alpha/2} e^{(2+\alpha) M_0})^{1+\alpha/2}  + 
    \Ma^{1+\alpha/2} \big)\nz{\uudel}\\
         & \peq \Ma^{7+3\alpha} e^{e^{2M_0+1}} e^{(3+2\alpha)M_0} \Dela
      + \Ma^{2+2\alpha} e^{(2+5\alpha/2) M_0} \nz{\uudel}.\\
    \end{split}
\end{equation}
Finally, due to inequality \eqref{equ:nz-3}, we get 
\begin{equation}\label{equ:ineq-hddd-udel-4}
    \nonumber 
    \begin{split}
      \hddd{\uudel}    
         & \peq \Ma^{7+3\alpha} e^{e^{2M_0+1}} e^{(3+2\alpha)M_0} \Dela
      + \Ma^{5+4\alpha} e^{(2+5\alpha/2) M_0} e^{e^{2M_0+1}} \Dela\\
      & \peq \Ma^{7+3\alpha} e^{e^{2(M_0+1)}}\Dela.
    \end{split}
\end{equation}
\end{proof}
\appendix
\section{Anisotropic H\" older Spaces}
\label{sec:appendix}
\subsection{Definitions and notation}
\label{sse:Holder}
Classical (anisotropic) H\" older spaces provide a convenient
 setting for stability analysis of a class of
utility-maximization problems.  Here is a short
overview of the notation and some basic definitions.

Let $\C$ be the set of all
continuous functions $u:[0,T]\times\R\to\R$, and let $\Cb$ be a
sub-class of $\C$ containing only bounded functions.  $\Cb$ is a
Banach space under the ``sup''-norm
\[ \nz{u}=\sup_{(t,x)\in [0,T]\times \R}\abs{u(t,x)}.\] In addition to
the ``vanilla'' norm $\nz{\cdot}$, we introduce a family of
equivalent, weighted, norms $\sets{\nb{\cdot}}{\beta\geq 0}$, given by
\[ \nb{u}=\sup_{(t,x)\in [0,T]\times \R} e^{-\beta(T-t)}
\abs{u(t,x)},\text{ for }\beta\geq 0.\] Due to the importance of
one-sided bounds, we use the following notation
\[ 
\s{u}=\sup_{(t,x)\in [0,T]\times \R} u(t,x),\ \i{u}=\inf_{(t,x)\in
  [0,T]\times \R} u(t,x),
\]
as well as their section-wise counterparts
\[ 
\s{u(t,\cdot)}=\sup_{x\in \R} \, u(t,x), \ 
\i{u(t,\cdot)}=\inf_{x\in \R}\,  u(t,x).
\]
The \define{parabolic distance} $d_p$ between $(t_1,x_1)$ and
$(t_2,x_2)$ in $[0,T]\times\R$ is defined by
\[d_p\Big((t_1,x_1),(t_2,x_2)\Big)=\sqrt{\abs{t_1-t_2}}+\abs{x_1-x_2}.\] For a
function $u\in \C$ and a constant $\alpha\in (0,1]$, we define its
$\alpha$-\define{H\" older constant} $\hd{u}\in [0,\infty]$ by
\begin{equation}
  \label{equ:u-delta}
  \begin{split}
    \hd{u}=\sup_{(t_1,x_1)\not= (t_2,x_2)\in [0,T]\times \R}\ \frac{
      \abs{u(t_1,x_1)-u(t_2,x_2)}}{
d_p\Big((t_1,x_1),(t_2,x_2)\Big)^{\alpha}}.
  \end{split}
\end{equation}
The functional $\nzd{\cdot}$, given by
\begin{equation}
  \label{equ:hd-norm}
  \begin{split}
    \nzd{u}= \nz{u}+\hd{u},
  \end{split}
\end{equation}
is a norm and it 
turns the class $\Cd$ of all functions $u\in \Cb$ for which
$\hd{u}<\infty$ into a  Banach space.

For $k\in\N$, the space $C^k([0,T]\times \R)$ contains all functions
$u\in \C$ such that the partial derivatives $
\tfrac{\partial^{m+n}}{\partial t^m\partial x^n}u$ exist and are
continuous on $(0,T)\times \R$, for all $m,n\in\N_0$ such that
$2m+n\leq k$.  For $k\in\N$ and $\alpha\in (0,1]$ we introduce the
space $\Ckd$ by
\begin{equation}
  \label{equ:Ckd}
  \begin{split}
    \Ckd=\{u\in C^k([0,T]\times \R)\,:\, & \tfrac{\partial^{m+n}}{\partial x^n\partial
      t^m}u 
    \text{ admit extensions in $\Cd$} \\
    &\text{ for all $m,n\in\N_0$ such that $2m+n\leq k$}\}.
  \end{split}
\end{equation}
The norm $\nzkd{\cdot}$, given by
\[ \nzkd{u}= \sum_{2m+n=k} \hd{\tfrac{\partial^{m+n}}{\partial
    x^n \partial t^m}u}+ \sum_{2m+n \leq k}
\nz{\tfrac{\partial^{m+n}}{\partial x^n \partial t^m}u} \, \text{ for
} u\in \Ckd,\] turns $\Ckd$ into a Banach space.  In particular, we
shall have occasion to use the spaces $\Cdd$ and $\Cddd$ with norms
\begin{equation}
  \nonumber 
  \begin{split}
    \nzdd{u} & = \hd{u_x}+\nz{u}+ \nz{u_x},\text{ and }\\
    \nzddd{u} & = \hd{u_t}+\hd{u_{xx}}+\nz{u_{t}}+ \nz{u_{xx}}+\nz{u_{x}}+\nz{u}.
  \end{split}
\end{equation}
Analogous constructions can be performed in the \define{isotropic}
setting, i.e., in our case, for functions $u$ of a single
variable. For $\alpha\in (0,1]$, in an act of notation overload, we
set
\[ \hd{u}=\sup_{x_1\not= x_2\in\R}
\frac{u(x_1)-u(x_2)}{\abs{x_1-x_2}^{\alpha}},\
\ndr{u}=\hd{u}+\sup_{x\in\R} \abs{u(x)}.\] Then, the H\" older space
$\Ckdr$, $k\in\N_0$, is a linear space consisting of all functions
$u:\R\to\R$ which admit $k$ continuous derivatives and whose $k^{th}$
derivative $\tfrac{d^k}{dx^k}u$ satisfies
$\hd{\tfrac{d^k}{dx^k}u}<\infty$. It becomes a Banach space when
endowed with the norm $\nkdr{\cdot}$ defined in analogy with its
anisotropic counterpart.  Further information on H\" older spaces can
be found in a variety of classical treatments of parabolic PDEs; for
example, the reader might want to consult \cite{Kry96} for an
unbounded-domain setting similar to ours, or \cite{LadSolUra67} and
\cite{WuYinWan06} for a more thorough analysis of linear and
quasilinear parabolic PDEs on bounded domains.

  In addition to
the classical function spaces, we shall have occasion to use functions
which depend on an additional variable $n\in \set{0,1}$.  For
$\alpha\in (0,1]$, and $k\in\N_0$, the space of all such functions for
which $k+\alpha$ -- H\" older continuity is required on both $n=0$--
and $n=1$-- slices, will be denoted by $C^{k+\alpha}([0,T]\times
\R\times \set{0,1})$ (or $C^{k+\alpha}(\R\times \set{0,1})$ in the
isotropic case). The natural (and Banach) norm used in those spaces is
the maximum of the H\" older norms of the $n$-slices:
\begin{equation}
\nonumber 
   \begin{split}
 \nzkd{u}&=\max( \nzkd{u(\cdot,0)}, \nzkd{u(\cdot,1)}),\\
\hkd{u}&=\max( \hkd{u(\cdot,0)}, \hkd{u(\cdot,1)}).
   \end{split}
\end{equation}
Similarly, for $u\in \Cs$, we define
\[ 
\s{u}=\max(\s{u(\cdot,\cdot,0)}, \s{u(\cdot,\cdot,1)}),\text{ and }
\i{u}=\min(\i{u(\cdot,\cdot,0)}, \i{u(\cdot,\cdot,1)}).\]
\subsection{Some useful results on H\" older spaces}
\label{sse:convention}
When
dealing with  various constants in the statements and proofs
of results below, we use the following convention: 
\begin{convention}\ 
\label{con:constants}
\begin{enumerate}
\item The variables $\mu$, $\gamma$ and $\alpha\in (0,1)$ are
  considered ``global'' and will not change throughout the paper. Any
  function of the global variables (and global variables only) if
  called a \define{universal} constant.
\item The notation $a\preceq b$ means that there exists
  a universal constant $C>0$ such that $a\leq C b$.  Such a constant
  may change from line to line.
\end{enumerate}
\end{convention}
We start with several well-known interpolation results which we
rephrase (and minimally adjust).
\begin{theorem}[{\bf Parabolic interpolation - additive form} -
  \cite{Kry96}, Theorem 8.8.1., p.~124.]
  There exists a universal constant $C>0$, such
  that for any $\eps>0$, and $u\in\Cddd$ we have
  \label{thm:interpol}
  \begin{equation}
    \nonumber 
    \begin{split}
      \hd{u} &\leq \eps \hddd{u}+ C \eps^{-\alpha/2} \nz{u} , \\
      \hd{u_x} &\leq \eps \hddd{u}+ C \eps^{-(1+\alpha)} \nz{u} , \\
      \nz{u_x} &\leq \eps \hddd{u}+ C \eps^{-1/(1+\alpha)} \nz{u} . \\
    \end{split}
  \end{equation}
\end{theorem}
\begin{corollary}
\label{cor:onen}
There exists a universal constant $C>0$ such that any
function $u\in\Cddd$ which satisfies the inequality
\[ \hddd{u}\leq D+E\nz{u}+F\hd{u}+G\nz{u_x}+H\hd{u_x},\]
for some constants $D,E,F,G\geq 0$, also satisfies the inequality
\begin{equation}
   \label{equ:yy}
   \begin{split}
     \hddd{u}\leq C\Big(
     D+\Big[E+F^{1+\alpha/2}+G^{1+\tfrac{1}{1+\alpha}}+H^{2+\alpha}\Big]\nz{u}
     \Big).
   \end{split}
\end{equation}
\end{corollary}
\begin{proof}
  By Theorem \ref{thm:interpol}, with the choice of $\eps=\tfrac{1}{6F}$, 
we have $\hd{u}\leq \tfrac{1}{6F}
  \hddd{u}+ C (6F)^{\alpha/2}\nz{u}$, so that
\begin{equation}
\nonumber 
   \begin{split}
F \hd{u} & \leq \tfrac{1}{6} \hddd{u} + C F^{1+\alpha/2}\nz{u}.
   \end{split}
\end{equation}
Similarly, 
\begin{equation}
\nonumber 
   \begin{split}
G \nz{u_x} & \leq \tfrac{1}{6} \hddd{u} + C
G^{1+\tfrac{1}{1+\alpha}}\nz{u} \quad \text{and} \quad 
H \hd{u_x}  \leq \tfrac{1}{6} \hddd{u} + C H^{2+\alpha}\nz{u}.\\
   \end{split}
\end{equation}
The estimate \eqref{equ:yy} now follows from
\begin{equation}
\nonumber 
   \begin{split}
 \hddd{u}&\leq D+E\nz{u}+F\hd{u}+G\nzd{u_x}+H \hd{u_x} \\
&\leq D+ \tot \hddd{u}+
C\Big[E+F^{1+\alpha/2}+G^{1+\tfrac{1}{1+\alpha}}+H^{2+\alpha}\Big]\nz{u}.
   \end{split}
\end{equation}
\end{proof}
\begin{theorem}[{\bf Parabolic interpolation - multiplicative form} -
  \cite{Kry96}, Exercise 8.8.2., p.~125.]
  There exists a universal constant $C>0$ such that for any
  $u\in\Cddd$ we have
  \label{thm:mult-interpol}
  \begin{equation}
    \label{equ:mult-interpol}
    \begin{split}
      \hd{u}^{2+\alpha}&\leq C  \hddd{u}^{\alpha}
      \nz{u}^{2},\\
      \nz{u_x}^{2+\alpha}&\leq C  \hddd{u}
\nz{u}^{1+\alpha},\\
      \hd{u_x}^{2+\alpha}&\leq C  \hddd{u}^{1+\alpha}
     \nz{u}.\\
    \end{split}
  \end{equation}  
\end{theorem}
\begin{theorem}[{\bf A H\" older estimate} - based on \cite{Kry96}, Exercise
  9.1.4, p.~139.]
  \label{thm:Holder}
  There exists a universal constant $C>0$ such that
 for any $u\in\Cddd$
  \begin{equation}
   \label{equ:Holder-est}
   \begin{split}
 \hddd{u} &\leq C ( \hd{u_t+\tot u_{xx}}+\hddd{u(T,\cdot)}).\\
   \end{split}
\end{equation}
\end{theorem}
\begin{corollary}[{\bf A H\" older estimate with a transport term}]
\label{cor:hold-trans}
  There exists a universal constant $C>0$ such that for any
  $u\in\Cddd$ and  $h\in\Cd$  we have
\begin{equation}
   \label{equ:hold-trans}
   \begin{split}
     \hddd{u} &\leq C \Big( \hd{u_t+\tot u_{xx}+h
       u_x}+\hddd{u(T,\cdot)} +
     \big(1+\nzd{h}^{2+\alpha} \big)  \nz{u}\Big).
   \end{split}
\end{equation}
\end{corollary}
\begin{proof}
The H\" older estimate of
  Theorem \ref{thm:Holder} implies that for 
 $f=u_t+\tot u_{xx}+h u_x$ and $g=u(T,\cdot)$
 we have
  \begin{equation}
    \nonumber 
    \begin{split}
      \hddd{u}
      &\peq \hddd{g}+\hd{f-h u_x} 
\peq \hddd{g}+\hd{f}+\nz{h}\hd{u_x}+\hd{h}\nz{u_x},
\end{split}
\end{equation}
and  \eqref{equ:hold-trans} follows from Corollary \ref{cor:onen}.
\end{proof}
We finish this section with a preparatory result which states that 
composition with a $C^2$-function $x\mapsto e^{\gamma x}$ is a locally
Lipschitz mapping on $\Cd$. Even though it is quite likely that such a
result is well-known, we were unable to locate a reference, and so,
for the sake of completeness, a proof is provided.

\begin{lemma}
\label{lem:F-delta}
For $\gamma\geq 0$, let $\Eg:\Cd\to\C$ be the composition mapping
\[ \Eg u (t,x)=e^{\gamma u(t,x)} \text{ for }u\in\Cd.\] Then $\Eg
u\in\Cd$ and the following bounds hold for all $u,\uone,\utwo\in\Cd$,
\begin{align}
   \label{equ:est-comp-1}
 \nz{\Eg u} & \leq e^{\gamma \s{u}}, \\
   \label{equ:est-comp-2}
\hd{\Eg  u} & \leq \gamma e^{\gamma \s{u}} \hd{u}, \\
   \label{equ:est-comp-3}
   \nz{\Eg\utwo-\Eg\uone} & \leq \gamma D  \nz{\utwo-\uone}, \text{ and }\\
   \label{equ:est-comp-4}
 \hd{\Eg\utwo-\Eg\uone} & \leq \gamma D \Big( 
\hd{\utwo-\uone}+
\gamma \hd{\utwo} \nz{\utwo-\uone}\Big), 
\end{align}
where $D=e^{\gamma \s{\uone}\vee \s{\utwo}}$. 
\end{lemma}
\begin{proof}
  Since $0\leq \Eg u\leq \sup_{t,x} e^{\gamma u(t,x)}$,
  \eqref{equ:est-comp-1} holds. For \eqref{equ:est-comp-2}, we note
  that, by the intermediate value theorem, 
\begin{equation}
   \label{equ:qwe}
   \begin{split}
 \abs{e^{\gamma u(s,y)}-e^{\gamma u(t,x)}}=\gamma e^{\gamma \xi} \abs{u(s,y)-u(t,x)},
   \end{split}
\end{equation}
for some convex combination $\xi$ of $u(t,x)$ and $u(s,y)$, so that 
$e^{\gamma \xi}\leq e^{\gamma \s{u}}.$ 

In order to get the other two bounds, we pick $\uone,\utwo\in \Cd$,
$(t_1,x_1),$ $ (t_2,x_2)\in [0,T]\times \R$ and set
$d=d_p\Big((t_1,x_1),(t_2,x_2) \Big)$ and $\delta \ukey=
\ukey(t_2,x_2)-\ukey(t_1,x_1)$ for $k=1,2$.  Observe that we
can prove \eqref{equ:est-comp-3} in a similar manner as
\eqref{equ:est-comp-2}.  Focusing on \eqref{equ:est-comp-4}, we set
\[\delta=e^{\gamma \utwo(t_2,x_2)}-
e^{\gamma\uone(t_2,x_2)}-(e^{\gamma \utwo(t_1,x_1)}-e^{\gamma \uone(t_1,x_1)})\]
and note the 
elementary equality \[G(b)-G(a)=(b-a) \int_0^1 G'(h^\theta(a,b))\,
d\theta,\] where $h^\theta(a,b)=(1-\theta) a+\theta b$, $G\in C^1(\R)$. Then, 
\begin{equation}
    \nonumber 
    \begin{split}
      |\delta|&= \Big| \delta \utwo \int_0^1 \gamma e^{\gamma
        h^\theta(\utwo(t_1,x_1),\utwo(t_2,x_2))}\, d\theta  -\delta \uone \int_0^1 \gamma e^{\gamma
        h^\theta(\uone(t_1,x_1),\uone(t_2,x_2))}\, d\theta\Big|
      \\
      &\leq \gamma D \Big|\delta \utwo-\delta \uone\Big| + \Big|
    \delta \utwo \int_0^1 \delta h(\theta)
    \int_0^1 \gamma^2 e^{\gamma g(\eta,\theta) }\, d\eta\,
    d\theta\Big|,
    \end{split}
\end{equation}
where \begin{equation}
\nonumber 
   \begin{split}
 \delta
h(\theta)&=h^\theta(\utwo(t_1,x_1)-\uone(t_1,x_1),\utwo(t_2,x_2)-\uone(t_2,x_2)),\\
g(\eta,\theta)&= h^{\eta}( h^{\theta}(\uone(t_1,x_1),\uone(t_2,x_2)),
h^{\theta}(\utwo(t_1,x_1),\utwo(t_2,x_2))).
   \end{split}
\end{equation}
 Clearly,
$g(\eta,\theta)$ is a convex combination of values of functions
$\uone$ and $\utwo$ and the inequality
$\delta h(\theta)
\leq \nz{\utwo-\uone}$ holds. Therefore,
\begin{equation}%
\label{equ:second-wwe}
    \begin{split}
      |\delta|&\leq D \Big( 
\hd{\utwo-\uone}  \gamma d^{\alpha}+
\hd{\utwo} \nz{\utwo-\uone} \gamma^2 d^{\alpha}\Big),
    \end{split}
\end{equation}
which directly implies \eqref{equ:est-comp-4}.
\end{proof}
\ifx \cprime \undefined \def \cprime {$\mathsurround=0pt '$}\fi\ifx \k
  \undefined \let \k = \c \fi\ifx \scr \undefined \let \scr = \cal \fi\ifx
  \soft \undefined \def \soft {\relax}\fi

\end{document}